\documentclass[]{interact}

\usepackage{xcolor}

\usepackage{epstopdf}

\usepackage{natbib}
\bibpunct[, ]{(}{)}{;}{a}{}{,}

\usepackage{amsfonts}
\usepackage{amsthm}
\usepackage{amsmath}

\DeclareMathOperator*{\argmax}{argmax} 

\DeclareMathOperator{\diag}{diag}
\DeclareMathOperator{\N}{N}

\usepackage{dsfont}
\usepackage{nicefrac}

\usepackage{tabularx}
\usepackage{dcolumn}
\newcolumntype{d}[1]{D{.}{.}{#1}}

\theoremstyle{plain}
\newtheorem{theorem}{Theorem}[section]

\newtheorem{corollary}[theorem]{Corollary}
\newtheorem{proposition}[theorem]{Proposition}

\theoremstyle{definition}

\theoremstyle{remark}
\newtheorem{remark}{Remark}

\usepackage{pgfplots}
\usepackage{pgfplotstable}
\pgfplotsset{compat=1.7}
\usetikzlibrary{calc}

\definecolor{color1}{HTML}{1f77b4}
\definecolor{color2}{HTML}{ff7f0e}
\definecolor{color3}{HTML}{2ca02c}
\definecolor{color4}{HTML}{d62728}
\definecolor{color5}{HTML}{9467bd}
\definecolor{color6}{HTML}{8c564b}
\definecolor{color7}{HTML}{e377c2}
\definecolor{color8}{HTML}{7f7f7f}

\pgfplotscreateplotcyclelist{proxylist}{%
{color1,mark=triangle},
{color2,mark=x},
{color3,mark=square},
{color4,mark=o},
{color5,mark=+},
{color6,mark=pentagon},
{color7,mark=star},
{color8,mark=diamond}
}

\newlength{\plotWidth}
\newlength{\plotHeight}
\newlength{\vSpace}
\newlength{\tempvSpace}
\newlength{\hSpace}
\tikzset{every node/.prefix style={font=\footnotesize}}

   \pgfplotsset{
      every axis/.prefix style = {
      tick label style = {font=\footnotesize},
      grid = both,
      scaled y ticks = false,
      yticklabel style = {/pgf/number format/fixed}},
      every axis legend/.prefix style = {font=\footnotesize},
      cycle list name = proxylist
    }
    
    \pgfplotsset{twoonpage/.style={height=\plotHeight,width=2\plotWidth}}
    \pgfplotsset{threeonpage/.style={height=\plotHeight,width=2\plotWidth}}
    \pgfplotsset{sixonpage/.style={height=0.75\plotHeight,width=\plotWidth}}
    \pgfplotsset{eightonpage/.style={height=0.75\plotHeight,width=\plotWidth}}
    
\pgfplotsset{RMSEstyle/.style={
      twoonpage,
      xtick = {0, 1, 2, 3, 4, 5, 6, 7, 8, 9, 10, 11, 12},
      xticklabels={21-08, 21-09, 21-10, 21-11, 21-12, 22-01, 22-02, 22-03, 22-04, 22-05, 22-06, 22-07, 22-08},
      xmin = 0,
      xmax = 12,
      every axis plot/.append style={dotted, mark options={solid, thick}},
      }
}

   \pgfplotsset{proxystyle/.style={
      x tick label style={font=\tiny,rotate=90},
      x tick label as interval,
      xtick = {0, 31, 61, 92, 122, 153, 184, 212, 243, 273, 304, 334, 365, 396, 426, 457, 487, 518, 549, 577, 608, 638, 669, 699, 730},
      xticklabels={20-08, 20-09, 20-10, 20-11, 20-12, 21-01, 21-02, 21-03, 21-04, 21-05, 21-06, 21-07,21-08, 21-09, 21-10, 21-11, 21-12, 22-01, 22-02, 22-03, 22-04, 22-05, 22-06, 22-07},
      xmin = 0,
      xmax = 730
   }} 

\begin{document}

\articletype{ARTICLE TEMPLATE}

\title{Bitcoin option pricing: A market attention approach}

\author{
\name{\'Alvaro Guinea Juli\'a\textsuperscript{a}\thanks{\'Alvaro Guinea Juli\'a. Email: agjulia@icai.comillas.edu}  and Alet Roux\textsuperscript{b}\thanks{Alet Roux. Email: alet.roux@york.ac.uk}}
\affil{\textsuperscript{a}Comillas Pontifical University ICAI, Madrid, 28015, Spain; \textsuperscript{b}University of York, Heslington, YO10 5DD, United Kingdom.}
}

\maketitle

\begin{abstract}
A model is proposed for Bitcoin prices that takes into account market attention. Market attention, modeled by a mean-reverting Cox-Ingersoll-Ross processes, affects the volatility of Bitcoin returns, with some delay. The model is affine and tractable, with closed formulae for the conditional characteristic functions with respect to both the conventional and a delayed filtration. This leads to semi-closed formulae for European call and put prices.  A maximum likelihood estimation procedure is provided, as well as a method for changing to a risk-neutral measure. The model compares very well against classical and attention-based models when tested on real data.
\end{abstract}


\section{Introduction}

The term ``market attention'' refers to the attention that investors or media pay to a particular stock or asset. It is also referred to investor attention or market interest. Market attention can be measured in different ways. Traditionally, trading volume, news coverage, or extreme past returns have been used as proxies for market attention and have been shown to affect stock prices \citep[cf.][]{hou2009tale,chen2020investor}. With the increase in the use of the Internet in the last two decades, new proxies for market attention have appeared, such as the number of Google searches or the Twitter volume. \cite{da2011search} show that the number of Google searches affects the prices of stocks in the Russell 3000 index. It has also been shown that Twitter sentiment affects the prices of the stocks in the Dow Jones Industrial Average index \citep{ranco2015effects}. Additionally, Twitter volume also affects option prices \citep{wei2016twitter}.

We are interested in building a model for pricing Bitcoin options that incorporates market attention, with particular emphasis on tractable (ideally, closed-form) option prices for vanilla options. Bitcoin is a highly liquid cryptocurrency first introduced by \cite{nakamoto2008bitcoin}, and in recent years, with the increase in the value of Bitcoin, new exchanges have emerged that offer Bitcoin options, for example, Deribit (\texttt{www.deribit.com}), LedgerX (\texttt{https://derivs.ftx.us}) and Bit (\texttt{www.bit.com}).

There is a rapidly increasing body of literature studying the relationship between cryptocurrencies and market attention. For example, \cite{smales2022investor} built a panel regression model to show that the number of Google searches generates greater returns and more volatility for the most important cryptocurrencies. \cite{eom2019bitcoin} constructed an autoregressive model for volatility that includes the number of Google searches and showed that the number of Google searches improves the predictability of Bitcoin volatility. Also, \cite{philippas2019media} constructed a dual process diffusion model to express that media attention (Google searches and Twitter volume) partially affects the prices of Bitcoin. \cite{aalborg2019can} constructed several linear regression models to show that trading volume affects Bitcoin volatility. \cite{al2021cryptocurrency} used a VAR model to prove that an increase in investor attention, the Twitter volume of the last five days, produces an increase in Bitcoin volatility. \cite{kristoufek2015main} used wavelet analysis to indicate that an increase in interest (number of Google searches and Wikipedia views) produces an increase in prices during bubble formation and a decrease in prices during bubble burst. Evidence of this phenomenon can also be found in the work of \cite{zhang2021time}, where the authors show, using the Granger causality test, that an increase in the number of Google searches contributes to bubble formation in Bitcoin prices. \cite{ciaian2016economics} used time series models to show that the number of views on Wikipedia also has an impact on the price of Bitcoin in the short term. Finally, \cite{aslanidis2021link} built a Google trends cryptocurrency index and showed that there is a short term bidirectional relation between the price of Bitcoin and the Google trends index.

It is clear from the above that market attention seems to affect Bitcoin volatility. There is a further strand in the literature that uses GARCH type models to analyze the relation between market attention and the volatility of Bitcoin. For example, \cite{lopez2021bitcoin} constructed a sentiment index for Bitcoin using the software Stanford Core NLP and the web page StockTwits (\texttt{https://stocktwits.com}). The authors then built GARCH type models and showed that the sentiment index affects the volatility of Bitcoin returns. In addition, using a GARCH-MIDAS model, \cite{liang2022predictor} showed that the number of Google searches has an impact on predicting the volatility of Bitcoin. \cite{figa2019does} used ARMA-GARCH models to show that trading volume affects the mean and the volatility of Bitcoin returns, while the number of Google searches only affects the volatility of Bitcoin returns. Finally,  \cite{figa2019disentangling} used a VAR-EGARCH model to show that neither the trading volume nor the number of Google searches affect the mean of Bitcoin returns. However, both attention proxies affect the volatility of Bitcoin returns. 

We observe that the results in the literature seem to disagree on some aspects. Some studies \citep{smales2022investor,figa2019does,kristoufek2015main,ciaian2016economics} claim that market attention influences the mean and volatility of Bitcoin returns, while others \citep{figa2019disentangling,aalborg2019can,al2021cryptocurrency,suardi2022predictive} conclude that market attention affects only the volatility of Bitcoin returns. In this paper, we take the latter approach by assuming that market attention affects the volatility of Bitcoin returns but not the mean.

In most of the econometric models mentioned above, the market attention is included in those models with a delay. That is, market attention does not act instantaneously with respect to the response variable (usually the return of Bitcoin or the volatility of Bitcoin returns). We take this into account by assuming that market attention acts on Bitcoin returns with a constant delay---the value of which can be estimated from the data.

Bitcoin option pricing has received some attention in the research literature. Due to the high volatility of Bitcoin prices, several models include jumps in the price structure \citep{shirvani2021bitcoin,chen2021detecting,olivares2020pricing}. \citet{cao2021valuation} propose an equilibrium model to price Bitcoin options. \citet{Cretarola_FigaTalamanca_Patacca2020} constructed a stochastic volatility model that incorporates market attention into the option pricing model. These models are all continuous in time, but discrete time models have also been used for pricing Bitcoin options \citep{venter2020price,siu2021bitcoin}. Bitcoin options are priced in US dollars but also in Bitcoin; \citet{alexander2021inverse} proposed a method to price options in number of Bitcoins (as offered on Deribit, for example) instead of a reference currency (usually US dollars).

The model proposed in this paper is related to the attention-based model of \citet{Cretarola_FigaTalamanca_Patacca2020}, who proposed a two-dimensional model similar to ours, with the main difference being that market interest is modeled as a geometric Brownian motion. In this paper, we model the interest as a Cox-Ingersoll-Ross process, which appears to fit well with the mean-reverting nature of a wide range of proxies for market attention and allows it to affect the volatility but not the drift. This means that our model is affine, hence analytically tractable.

The structure of the paper is as follows. In Section \ref{sec:ME}, we specify the model and establish a number of basic properties, such as conditional independence of the log returns and a Markov property. We propose a maximum likelihood estimation procedure in Section \ref{sec:Inference}. In Section \ref{OptionPricing}, we introduce the change-of-measure, derive a formula for the conditional characteristic function, and discuss the practicalities of option pricing. The methods are applied to real data in Section \ref{sec:RealData}. The paper concludes with Section \ref{sec:conclusion}, which provides an outlook on future research

\section{Model specification}\label{sec:ME}

In this section, we specify the model and establish a number of basic properties that are useful for pricing and parameter estimation. Throughout, we assume a finite (possibly large) model horizon $H>0$; we assume $t\in[0,H]$ unless specified otherwise. The model is based on a probability space $(\Omega,\mathcal{F}, \mathds{P})$ containing two independent Brownian motions $W^P=(W^P_t)$ and $W^I=(W^I_t)$. Let $\mathbb{F}^I=(\mathcal{F}^I_t)$ and $\mathbb{F}^P=(\mathcal{F}^P_t)$ be the filtrations generated by $W^I$ and $W^P$, respectively, assumed to satisfy the usual conditions. Define the \emph{general filtration} $\mathbb{F}=(\mathcal{F}_t)$ by 
\begin{equation} \label{eq:filtration}
    \mathcal{F}_t=\mathcal{F}^P_t\vee\mathcal{F}^I_t \text{ for all } t,
\end{equation}
where $\vee$ denotes the join of $\sigma$-fields.

The level of market attention or interest is modeled by means of a positive-valued \emph{interest process} $I=(I_t)_{t\in[-L,H]}$, where $L>0$. On the interval $[-L,0]$, the interest process is specified by means of a deterministic positive-valued function representing its most recent past values, and when $t\in(0,H]$, it follows Cox-Ingersoll-Ross dynamics \citep{Cox_Ingersoll_Ross1985}, in other words,
\begin{align}
 I_t&=\phi_I(t)\text{ for all } t \in [-L,0], \label{eq:I:past} \\
 dI_t &= a(b - I_t)dt + \sigma_I \sqrt{I_t} dW^I_t \text{ for all } t > 0, \label{eq:I:CIR}
\end{align}
where $b \in \mathds{R}$, $a,\sigma_I>0$ and $\phi_I : [-L,0] \to (0,\infty)$ is a continuous deterministic function. The stochastic differential equation \eqref{eq:I:CIR} has a strong solution for any value of $I_0=\phi_I(0)>0$. We furthermore impose the condition
\begin{equation}\label{Conds}
    \frac{2ab}{\sigma_I^2} \geq1,
\end{equation}
which ensures that $I_t > 0$ for all $t\geq0$ almost surely \citep[Theorem~2.27]
{gulisashvili2012analytically}.

The Bitcoin \emph{price process} is modelled by specifying the dynamics of the log price $X=(X_t)$ as
\begin{align}\label{SDEG}
    X_t &= x + \mu t + \sigma_P\int_0^t\sqrt{I_{s-\tau}}dW^P_s \text{ for all } t,
\end{align}
where $x,\mu \in \mathds{R}$, $\sigma_P>0$ and $\tau\in[0,\min\{L,H\}]$ is the fixed lag parameter. Observe that the interest does not appear in the drift of the log price; however, its square root is proportional to its volatility. The price process $P=(P_t)$ is defined as
\[
 P_t = e^{X_t} = \exp\left\{x + \mu t + \sigma_P\int_0^t\sqrt{I_{s-\tau}}dW^P_s\right\} \text{ for all } t,
\]
and by the It\^o formula satisfies the stochastic differential equation and initial condition
\begin{equation}
    dP_t = \left(\mu + \tfrac{1}{2}\sigma_P^2I_{t-\tau} \right)P_tdt + \sigma_P \sqrt{I_{t-\tau}} P_tdW^P_t  \label{Price1}
\end{equation}
and initial condition $P_0=e^x \in (0,\infty)$.

Finally, we assume a fixed risk-free interest rate of $r\geq0$ in the reference currency.

\subsection{Properties of the log returns}\label{sec:log-returns}

The independence between the interest process $I$ and the Brownian motion~$W^P$ that drives the price process is key to the tractability of the log returns. The log return over any time period $[s,t]$ where $0\le s<t<H$ is defined as
\begin{equation}
 R_{s,t} = X_t-X_s = \mu(t-s) + \int_s^t \sigma_P \sqrt{I_{u-\tau}}dW^P_u, \label{eq:def:returns}
\end{equation}
with the final equality due to \eqref{SDEG}. The following remark follows from standard results \cite[cf.][Proposition 1.4.2]{GuineaJulia2022}.

\begin{remark}\label{RemarkNormal} 
Let $W=(W_t)$ be a Brownian motion with natural filtration $\mathbb{F}^W=(\mathcal{F}^W_u)$. Let $Z=(Z_t)$ be a left-continuous process independent of $W$ such that
$
 \int_0^T Z^2_udu<\infty
$ almost surely for some $T\in(0,H]$, and let $\mathbb{G}=(\mathcal{G}_t)$ be any filtration independent of $W$ with respect to which $Z$ is adapted. Then
\begin{align*}
\left.\int_s^T Z_udW_u\middle|\mathcal{G}_T \right. & \sim \N\left(0,\int_s^T Z_u^2du\right), \\
\left.\int_s^T Z_udW_u\middle|\mathcal{G}_T \vee \mathcal{F}^W_s\right. & \sim  \N\left(0,\int_s^T Z_u^2du\right)
\end{align*}
for all $s\in[0,T)$.
\end{remark}

The conditional normality of the integrals in the remark above leads to a number of important properties of the log returns. They have a conditional normal distribution, and moreover the following result shows that, conditional on the interest process, log returns over non-overlapping time periods are independent. The conditional independence will be very useful for estimating the parameters of the price process, particularly since the interest itself is observed. 

\begin{proposition}\label{prop:cond-independence}
For any $n\in\mathbb{N}$ and $0\le t_0 < \cdots < t_n \le H$, the log returns $R_{t_0,t_1},\ldots, R_{t_{n-1},t_n}$ are independent given $\mathcal{F}^I_H$, and moreover 
\[
 \left.\left(R_{t_0,t_1}, \ldots, R_{t_{n-1},t_n}\right)\middle|\mathcal{F}^I_H\right. 
 \sim N\left(\mu(t_1-t_0, \ldots, t_n-t_{n-1}\right), \sigma_P^2\diag\left(J_1,\ldots,J_n)\right),
\]
where \[
  J_j = \int_{t_{j-1}}^{t_j} I_{u-\tau} du \text{ for all } j=1,\ldots,n.
 \]
\end{proposition}

\begin{proof} 
By Remark~\ref{RemarkNormal} and the independence between $\mathcal{F}^I_H$ and $W^P$, the return $R_{s,t}$ for any $0\le s < t\le H+\tau$ has the same conditional distribution given $\mathcal{F}^I_H$ and $\mathcal{F}^P_s \vee \mathcal{F}^I_H$, namely $\N\left(\mu(t-s),\sigma_P^2\int_s^t I_{u-\tau} du\right)$. This means that
\begin{align} \label{eq:prop:cond-independence}
E\left[e^{i\lambda R_{s,t}}\middle|\mathcal{F}^I_H \right] &= E\left[e^{i\lambda R_{s,t}}\middle|\mathcal{F}^P_s \vee \mathcal{F}^I_H \right] = e^{i\lambda\mu(t-s) - \frac{1}{2}\lambda^2\sigma_P^2\int_s^t I_{u-\tau} du}
\end{align}
for all $\lambda\in\mathbb{R}$. Observe from the tower property of conditional expectation and \eqref{eq:prop:cond-independence} that, for any $\lambda_1,\lambda_2\in\mathbb{R}$,
\begin{align*}
 E\left[e^{i\left(\lambda_1R_{t_0,t_1}+\lambda_2R_{t_1,t_2}\right)}\middle|\mathcal{F}^I_H \right]
 &= E\left[e^{i\lambda_1R_{t_0,t_1}} E\left[e^{i\lambda_2 R_{t_1,t_2}}\middle|\mathcal{F}^P_{t_1} \vee \mathcal{F}^I_H \right] \middle|\mathcal{F}^I_H \right] \\
 &= \prod_{j=1}^2E\left[e^{i\lambda_j R_{t_{j-1},t_j}}\middle|\mathcal{F}^I_H \right] \\
 &= \prod_{j=1}^2e^{i\lambda_j\mu(t_j-t_{j-1}) - \frac{1}{2}\lambda_j^2\sigma_P^2J_j}.
\end{align*}
It is then straightforward to prove by induction that
\begin{align*}
  E\left[e^{i\sum_{j=1}^n\lambda_j R_{t_{j-1},t_j}}\middle|\mathcal{F}^I_H \right] 
  &= \prod_{j=1}^n E\left[e^{i\lambda_j R_{t_{j-1},t_j}}\middle|\mathcal{F}^I_H \right] \\
  &=
  \prod_{j=1}^ne^{i\lambda_j\mu(t_j-t_{j-1}) - \frac{1}{2}\lambda_j^2\sigma_P^2J_j}
\end{align*}
for all $\lambda_1,\ldots,\lambda_n\in\mathbb{R}$. The first equality is necessary and sufficient for the independence of $R_{t_0,t_1},\ldots, R_{t_{n-1},t_n}$ \citep[Theorem~4.1]{yuan2016some}. The second equality gives the characteristic function of the claimed conditional distribution.
\end{proof}

The following result will be useful in estimating the parameters of the log price from discretely observed data.

\begin{corollary} \label{corr:multinormal}
In the setting of Proposition \ref{prop:cond-independence}, define $J=(J_1,\ldots,J_n)$. Then
\[
 \left.R_{t_0,t_1}, \ldots, R_{t_{n-1},t_n}\middle|J\right. \sim N\left(\mu(t_1-t_0, \ldots, t_n-t_{n-1}\right), \sigma_P^2\diag\left(J_1,\ldots,J_n)\right).
\]
\end{corollary}

\begin{proof}
 The tower property of conditional expectation and Proposition \ref{prop:cond-independence} give for any $\lambda_1,\ldots,\lambda_n\in\mathbb{R}$ that
 \begin{align*}
  E\left[e^{i\sum_{j=1}^n\lambda_j R_{t_{j-1},t_j}}\middle|J \right] 
  &= E\left[E\left[e^{i\sum_{j=1}^n\lambda_j R_{t_{j-1},t_j}}\middle|\mathcal{F}^I_H \right]\middle|J \right] \\
  &= E\left[\prod_{j=1}^ne^{i\lambda_j\mu(t_j-t_{j-1}) - \frac{1}{2}\lambda_j^2\sigma_P^2J_j}\middle|J \right] \\
  &= \prod_{j=1}^ne^{i\lambda_j\mu(t_j-t_{j-1}) - \frac{1}{2}\lambda_j^2\sigma_P^2J_j}.
 \end{align*}
 This is the characteristic function of the claimed distribution.
 \end{proof}

\subsection{Delay and the Markov property}\label{sec:Markov}

In this section, we work with the log price process $X$ for convenience; however, the same considerations apply to the price process $P$, being a continuous transformation of $X$. Standard arguments from the definition of the It\^o integral  \cite[cf.][Theorem 1.2.1]{GuineaJulia2022} lead to the following.

\begin{remark}\label{ChangeBM} 
Let $W=(W_t)$ be a Brownian motion with natural filtration $\mathbb{F}^W=(\mathcal{F}^W_t)$, and define the Brownian motion $B=(B_t)_{t\in[0,H-\tau]}$ as 
\begin{equation} \label{eq:def:B}
B_t=W_{t+\tau}-W_\tau \text{ for all }t\in[0,H-\tau].
\end{equation}
Let $Z=(Z_T)$ be a left-continuous $\mathbb{F}^W$-adapted process such that $\int_0^TZ^2_{t-\tau}dt<\infty$ almost surely for some $T\in[0,H]$, and define $Z_t=\phi(t)$ for all $t\in[-\tau,0]$ for some left-continuous function $\phi:[-\tau,0]\rightarrow\mathds{R}$. Then
\[
\int_{\tau}^t Z_{s-\tau} dW_s = \int_0^{t-\tau} Z_s dB_s \text{ for all } t\in(\tau,T].
\]
\end{remark}

The following simple relationship holds between the filtrations generated by the Brownian motions in Remark \ref{ChangeBM}.

\begin{proposition} \label{prop:FP-FB-Fd}
In the setting of Remark \ref{ChangeBM}, let $\mathbb{F}^B=(\mathcal{F}^B_t)_{t\in[0,H-\tau]}$ be the natural filtration of the Brownian motion $B$ defined in \eqref{eq:def:B}. Then
\begin{equation} \label{eq:filtration-equality}
\mathcal{F}^W_t = \mathcal{F}^W_\tau \vee \mathcal{F}^B_{t-\tau} \text{ for all }t\ge\tau.
\end{equation}
\end{proposition}

\begin{proof}
 On the one hand, it follows from \eqref{eq:def:B} and the properties of the filtration $\mathbb{F}^W$ that $B^W_{t-\tau}$ is $\mathcal{F}^W_t$-measurable and moreover that
$
\mathcal{F}^B_{t-\tau} \vee \mathcal{F}^W_\tau \subseteq \mathcal{F}^W_t.
$
On the other hand, \eqref{eq:def:B} gives that $W^W_t = B^W_{t-\tau}+W^W_\tau$, with the terms on the right hand side being independent. Thus
$
\mathcal{F}^W_t \subseteq \mathcal{F}^B_{t-\tau} \vee \mathcal{F}^W_\tau,
$
which concludes the proof.
\end{proof}

The log price process $X$ and the interest process $I$ are adapted to the general filtration $\mathbb{F}$ defined in \eqref{eq:filtration}. When $\tau>0$, then the joint process $(X,I)$ is generally not a Markov process with respect to this filtration; however, taking the delay into account when defining the filtration does lead to a Markov process. To this end, define the \emph{delayed filtration} $\mathbb{F}^{\mathrm{d}} = (\mathcal{F}^{\mathrm{d}}_t)$ as 
\begin{equation} \label{eq:delayed-filtration}
    \mathcal{F}^{\mathrm{d}}_t =
     \begin{cases}
                  \mathcal{F}^P_t  &\text{if } t < \tau,\\
                  \mathcal{F}^P_t  \vee \mathcal{F}^I_{t-\tau} &\text{if } t \ge \tau. \\
         \end{cases}
\end{equation}
Notice that $\mathcal{F}^{\mathrm{d}}_t \subseteq \mathcal{F}_t\subseteq \mathcal{F}$ for all $t$, and that the processes $X$ and $I_{\cdot-\tau}=(I_{t-\tau})$ are adapted with respect to $\mathbb{F}^{\mathrm{d}}$.

We have the following result.

\begin{proposition}\label{Markov Property}
The two-dimensional process $(X,I_{\cdot-\tau})$ is a Markov process with respect to the filtration $\mathbb{F}^{\mathrm{d}}$.
\end{proposition}

\begin{proof}
The proof relies on the fact that the strong solution of a (possibly multi-dimensional) time-homogeneous stochastic differential equation involving Brownian motion is a Markov process \cite[cf.][Theorem 5.4.20]{Karatzas}. For $t\le\tau$ this follows from \eqref{eq:I:past} and
\begin{align*}
X_t = X_0 +  \mu t +  \sigma_P \int_0^t \sqrt{\phi_I(u-\tau)}dW^P_u \text{ for all }t\in[0,\tau].
\end{align*}
For $t>\tau$, we have
\begin{align*}
X_t 
&= X_{\tau} + \mu(t-\tau) + \int_0^{t-\tau} \sigma_P \sqrt{I_u} dB^P_u
\end{align*}
by Remark \ref{ChangeBM}, where the Brownian motion $B^P=(B^P_t)_{t\in[0,H-\tau]}$ is defined as
\[B^P_t=W^P_{t+\tau}-W^P_\tau \text{ for all }t\in[0,H-\tau].\]
The independence of $W^P$ and $W^I$ means that $B^P$ and~$W^I$ are independent. Let $\mathbb{F}^B = ( \mathcal{F}_t^B)$ be the filtration generated by $B^P$. Define a new process $Z=(Z_t)_{t\in[0,H-\tau]}$ as
\begin{equation}\label{ProcessZMarkov}
Z_t = X_{t+\tau} = X_{\tau}  + \mu t +  \sigma_P \int_0^t \sqrt{I_u} dB^P_u \text{ for all } t\in[0,H-\tau].
\end{equation}
Observing that $\int_0^t \sqrt{I_u} dB^P_u$ is independent of $Z_0=X_{\tau}$ for all $t\in[0,H-\tau]$, it follows that $(Z,I)$ is a Markov process with respect to the filtration $\mathbb{F}^Z=(\mathcal{F}^{ZI}_t)_{t\in[0,H-\tau]}$ generated by it. For all $t\in[0,H-\tau]$ we have
\begin{align*}
\mathcal{F}^{ZI}_t = \mathcal{F}^P_{\tau} \vee \mathcal{F}_t^B  \vee \mathcal{F}^I_t = \mathcal{F}^P_{t+\tau} \vee \mathcal{F}^I_t = \mathcal{F}^{\mathrm{d}}_{t+\tau}
\end{align*}
by Proposition \ref{prop:FP-FB-Fd}.  Noting that $ (X_t,I_{t-\tau}) = (Z_{t-\tau},I_{t-\tau})$ for all $t>\tau$ completes the proof.
\end{proof}

\section{Parameter estimation}\label{sec:Inference}

In this section, we propose a maximum likelihood method for estimating the parameters of the processes $X$ and $I$ in \eqref{eq:I:past}--\eqref{eq:I:CIR} and \eqref{Price1} based on discrete observations. It is possible to estimate all the parameters of \eqref{eq:I:CIR} and \eqref{Price1} at once, using the maximum likelihood estimation method. As the model is affine, it is possible to compute the joint characteristic function of the random variables $X_t$ and $I_{t-\tau}$, and then recover the density function. The difficulty with such an approach is the need to compute a double integral; for this reason a different approach is taken here.

The function $\phi_I$ in \eqref{eq:I:past} can be approximated by interpolation. We adopt a multi-step approach in estimating the parameters of the system \eqref{eq:I:CIR} and \eqref{Price1}. Section \ref{EST_INT} focuses on the estimation of the parameters of the interest process; in the case of the deterministic function $\phi_I$ in \eqref{eq:I:past} this is done by approximation, and for the parameters in \eqref{eq:I:CIR} we apply a standard maximum likelihood procedure. The conditional likelihood method is used to estimate the parameters in \eqref{Price1}: in Section \ref{CondMLE}, we assume the lag parameter to be fixed and derive conditional maximum likelihood parameters for the drift and volatility of the log price. The estimation of the lag parameter will be transformed into a model selection problem; this is finally presented in Section \ref{tauCL}.

Data is assumed to be observed up to a finite time horizon $H>0$, and there are $N+1\in \mathds{N}$ observations. For simplicity of exposition the observations are assumed to be equispaced, with a time step $\Delta=\frac{H}{N}$ between observations. We also assume that $L=M\Delta$ for some $M \in \mathds{N}$ with $N \gg M$, and that $\tau = k\Delta$ for some $k=0,\ldots,M$---these assumptions are natural in that it will not be possible to estimate the lag parameter $\tau$ to greater precision than the available data.

\subsection{Interest process} \label{EST_INT}

The parameters of the interest process can be estimated independently of all other parameters. Let $y_j$ be the observed value of $I_{j\Delta}$ for all $j=-M,\ldots,N$. 

The observations $y_{-M},\ldots,y_0$ can be used to construct an approximation $\hat{\psi}_I:[-L,0]\rightarrow(0,\infty)$ for the deterministic function $\phi_I$ of \eqref{eq:I:past}. Clearly
\[
 \hat{\psi}_I(j\Delta) = y_j \text{ for all } j=-M,\ldots, 0,
\]
and the values of $\hat{\psi}_I$ at intermediate values can be determined using any suitable interpolation method. The numerical procedure in Section \ref{sec:RealData} uses linear interpolation, in other words,
\begin{multline}\label{PhiIFunct}
\hat{\psi}_I(t)  =  y_j + (t- j\Delta)\frac{y_{j+1}-y_j}{(j+1)\Delta-j\Delta}\text{ for all } t \in [j\Delta,(j+1)\Delta],\\ j=-M,\ldots, -1.  
\end{multline}

The observations $y_0,\ldots,y_M$ can be used to estimate the parameters appearing in \eqref{eq:I:CIR}. The maximum likelihood procedure described below is well established, with computer implementations by \citet{Kladivko2007} and \citet[Section~3.1.3]{Iacus2008}, among others.

Since the stochastic process $I$ has the Markov property \cite[Theorem 5.4.20]{Karatzas} and $I_0=\phi_I(0)$, the likelihood function is
\begin{align}\label{CH2:CIR_LIKELIHOOD}
L^I(a,b,\sigma_I) 
&=  f^I_{1:N}(y_1,\ldots,y_N|y_0,a,b,\sigma_I) \nonumber \\
&= \prod_{j=0}^{N-1} f^I_{j+1|j}(y_{j+1}|y_j,a,b,\sigma_I),
\end{align}
where $f^I_{1:N}$ is the joint density function of $I_{\Delta},\ldots,I_{N\Delta}$ given $I_0$ and $f^I_{j+1|j}$ is the conditional density function of $I_{(j+1)\Delta}$ given $I_{j\Delta}$ for $j=0,\ldots,N-1$. This conditional density is also referred to as the transition density of $I$.

The key idea, credited to \citet{Feller1951}, is that the transition density of the Cox-Ingersoll-Ross model can be expressed in terms of the density function of the well-known non-central chi-squared distribution. To be precise, we have for all $j=0,\ldots,N-1$ that
\begin{equation} \label{CH2:DENCIR}
  f^I_{j+1|j}(y_{j+1}|y_j,a,b,\sigma_I) =2c f_{\chi^2}(p_{j+1}|2(q+1),2u_j),
\end{equation}
where $f_{\chi^2}(\cdot|d,\lambda)$ is the probability density function of the non-central chi-squared distribution with $d>0$ degrees of freedom and non-centrality parameter~\mbox{$\lambda>0$}, and where
\begin{align*}
    c & = \frac{2a}{\sigma_I^2\left(1-e^{-a\Delta}\right)}, & q & = \frac{2 a b}{\sigma_I^2}-1, &
    u_j & = c y_j e^{-a \Delta }, &
    p_{j+1} & = 2c y_{j+1}. 
\end{align*}
Combining \eqref{CH2:CIR_LIKELIHOOD} and \eqref{CH2:DENCIR}, the log-likelihood function becomes
\[
l^I(a,b,\sigma_I) = \ln L^I(a,b,\sigma_I) = \sum_{j=0}^{N-1} \ln 2c f_{\chi^2}(p_{j+1}|2(q+1),2u_j).
\]
Numerical procedures can then be used to maximize the log-likelihood $l^I$ in order to obtain maximum likelihood estimators for $a$, $b$ and $\sigma_I$. 

\subsection{Drift and volatility of the log price}\label{CondMLE}

In this section, the lag parameter $\tau$ is assumed given, and the focus is on estimating the parameters $\mu$ and $\sigma_P$ in \eqref{Price1} by applying the conditional maximum likelihood method \cite[cf.][]{Andersen1970}.

The key idea behind the estimation procedure comes from Proposition~\ref{prop:cond-independence} via Corollary \ref{corr:multinormal}, which ensures that, conditional on the interest process, the log returns $R_{(j-1)\Delta,j\Delta}$ are independent normal random variables. For this reason it is natural to use the observed log returns for estimation. To this end, let $x_j$ be the observed value of the log price $X_{j\Delta}$ for all $j=0,\ldots,N$. The observed value of the log return $R_{(j-1)\Delta,j\Delta}$ for $j=1,\ldots,N$ is denoted as
\[
 r_j = x_j-x_{j-1} \text{ for } j=1,2,\ldots N.
\]

Define also
\begin{align} \label{eq:J}
J_j = \int_{(j-1)\Delta}^{j\Delta} I_{u-\tau}du \text{ for } j=1,2,\ldots N.
\end{align}
Let $z_j$ be the observed value of $J_j$ for all $j=1,\ldots,N$; observe that $J_j$ is deterministic for $j\le k$. As the interest process $I$ is not observed continuously, it is generally necessary to approximate the Lebesgue integrals $J_1,\ldots,J_N$, for example by the trapezoid method, in other words,
\begin{align}
z_j & = \tfrac{\Delta}{2}(y_{j-k}+y_{j-1-k}) \text{ for } j=1,2,\ldots N.
\label{eq:J-estimate}
\end{align}

By Corollary \ref{corr:multinormal}, the conditional joint density function of the log returns $R_{0,\Delta},\ldots,R_{(N-1)\Delta,N\Delta}$ given $J=(J_1,\ldots,J_N)$ is
\[
 f^{R|J}_{1:N}\left(r_1,\ldots,r_N\middle|z_1,\ldots,z_N,\mu,\sigma_P,\tau\right) = \prod_{j=1}^Nf_{\N}(r_j|\mu\Delta, \sigma_P^2z_j).    
\]
It follows that the conditional likelihood function is
\[
 L^{R|J}(\mu,\sigma_P,|\tau) = \prod_{j=1}^Nf_{\N}(r_j|\mu\Delta, \sigma_P^2z_j),
\]
and the conditional log-likelihood is
\begin{align}
 l^{R|J}(\mu,\sigma_P,|\tau) = \ln L^{R|J}(\mu,\sigma_P,|\tau) = \sum_{j=1}^N\ln f_{\N}(r_j|\mu\Delta, \sigma_P^2z_j). \label{LikeliHoodFunc1}
\end{align}
Maximising the conditional log-likelihood leads to analytic expressions for the maximum likelihood estimators of $\mu$ and $\sigma_P$.

\begin{proposition} \label{prop:MLE-musigma}
The conditional maximum likelihood estimators of $\mu$ and $\sigma_P$ are 
\begin{align} \label{eq:MLE-musigma}
\hat{\mu} &= \frac{\sum_{j=1}^N \frac{r_j}{z_j}}{\Delta \sum_{j=1}^N \frac{1}{z_j}}, &
\hat{\sigma}_P &= \sqrt{\frac{1}{N}\sum_{j=1}^N\frac{\left( r_j - \hat{\mu} \Delta\right)^2}{z_j}}.
\end{align}
\end{proposition}

\begin{proof}
Expanding \eqref{LikeliHoodFunc1} gives
\begin{align*}
 l^{R|J}(\mu,\sigma_P|\tau) = -N \ln\sigma_P  -\sum_{j=1}^N \ln\sqrt{2\pi z_j} - \frac{1}{2\sigma_P^2}\sum_{j=1}^N \frac{(r_j - \mu \Delta)^2}{z_j},
\end{align*}
and consequently
  \begin{align*}
\frac{\partial{l^{R|J}}}{\partial \mu}(\mu,\sigma_P|\tau)
 &= \frac{\Delta}{\sigma_P^2}\sum_{j=1}^N \frac{r_j - \mu \Delta}{z_j},\\
 \frac{\partial{l^{R|J}}}{\partial \sigma_P}(\mu,\sigma_P|\tau)
 &= \frac{1}{\sigma_P^3}\sum_{j=1}^N \frac{(r_j - \mu \Delta)^2}{z_j} -\frac{N}{\sigma_P}.
 \end{align*}
 Setting both derivatives equal to $0$ and solving for $\mu$ and $\sigma_P$ yields \eqref{eq:MLE-musigma}.
 
 In order to confirm that $\hat{\mu}$ and $\hat{\sigma}_P$ maximize the conditional log-likelihood, we derive its Hessian, namely,
  \[
\begin{bmatrix} \frac{\partial^2}{\partial \mu^2} & \frac{\partial^2}{\partial \mu \partial \sigma_P}\\
\frac{\partial^2}{\partial \sigma_P \partial\mu} & \frac{\partial^2}{\partial \sigma_P^2}
\end{bmatrix}l^{R|J}(\mu,\sigma_P|\tau)
    = \begin{bmatrix}
-\frac{\Delta^2}{\sigma_P^2} \sum_{j=1}^N\frac{1}{z_j} & -\frac{2\Delta}{\sigma_P^3}\sum_{j=1}^N \frac{r_j - \mu \Delta}{z_j} \\
-\frac{2\Delta}{\sigma_P^3}\sum_{j=1}^N \frac{r_j - \mu \Delta}{z_j} & \frac{N}{\sigma_P^2} - \frac{3}{\sigma_P^4}\sum_{j=1}^N \frac{(r_j - \mu \Delta)^2}{z_j} 
\end{bmatrix}.
\]
Evaluating the Hessian at $(\hat{\mu},\hat{\sigma}_P)$ yields
  \[
\begin{bmatrix} \frac{\partial^2}{\partial \mu^2} & \frac{\partial^2}{\partial \mu \partial \sigma_P}\\
\frac{\partial^2}{\partial \sigma_P \partial\mu} & \frac{\partial^2}{\partial \sigma_P^2}
\end{bmatrix}l^{R|J}(\hat{\mu},\hat{\sigma}_P|\tau)
    = \begin{bmatrix}
-\frac{\Delta^2}{\hat{\sigma}_P^2} \sum_{j=1}^N\frac{1}{z_j} & 0 \\
0 & -\frac{2N}{\hat{\sigma}_P^2} 
\end{bmatrix}.
\]
As this is negative definite, it follows that $\hat{\mu}$ and $\hat{\sigma}_P$ maximize the conditional log-likelihood as claimed.
\end{proof}



\subsection{Lag parameter}\label{tauCL}

The parameter $\tau$ takes one of the values $0,\Delta,2\Delta,\ldots,M\Delta =L$, and therefore the estimation of $\tau$ can be reduced to a model selection problem. That is, we select among the different models for the log price, where each model is formulated as
\begin{equation*}
    X_t = X_0 +  \int_0^t \mu du  + \int_0^t \sigma_P \sqrt{I_{u-\ell\Delta}} dW^P_u \text{ for all }t,
\end{equation*}
for some $\ell = 0,1,\ldots,M$.

As the process $I$ is given, the conditional likelihood function defined in Section \ref{CondMLE} can be used. Model selection takes place by minimizing an appropriate information criterion, such as the Akaike information criterion
\[
    \text{AIC}_r = 2q -2 l^{R|J}\left(\hat{\mu},\hat{\sigma}_P \left| \ell \Delta \right.\right) \text{ for } r = 0,\ldots, M 
\]
\cite[p.~605]{deleeuw1992introduction}, where $q$ is the number of parameters in the model, or the Bayesian information criterion
\[
    \text{BIC}_r = q\ln N -2 l^{R|J}\left(\hat{\mu},\hat{\sigma}_P \left| \ell \Delta \right.\right) \text{ for } r = 0,\ldots, M 
\]
\citep{neath2012bayesian}. In this case, as the number of parameters is known and fixed, these two criteria are equivalent, and estimates of $\tau$ are obtained by setting $\hat{\tau}=\hat{\ell}\Delta$, where
\[
\hat{\ell} \in \argmax_{\ell=0,\ldots,M} l^{R|J}\left(\hat{\mu},\hat{\sigma}_P \left| \ell \Delta \right.\right).
\]
Whilst in theory these estimates need not be unique, in practice this leads to a unique estimate (even if only through numerical rounding).

\subsection{Numerical experiment}

Numerical experiments can be used to evaluate the effectiveness of the estimation procedure described above. To this end, we simulate the interest and price processes with parameter values that are similar to those observed in applications to real data (see Figures and \ref{fig:Bitcoin-prices}, \ref{fig:proxies} and \ref{fig:InferenceResults}. We take 
\begin{align*}
 a &= 30, & b &= 15, & \sigma_I &= 0.6, & \mu &= 0, & \sigma_P &= 0.2, & \tau &= 0.025
\end{align*}
together with initial values $P_0=20,000$ and $\phi_I(t)=14$ for all $t\in[-L,0]$, where $L=2\tau=0.05$. We fix $\Delta=\nicefrac{1}{360}$ (observe that $\tau=9\Delta$) and perform the estimation procedure above to estimate the parameters of the model, for all time horizons of the form $T = 9n\Delta$ where $n=1,2,\ldots,80$. This setting reflects daily data for a range of time horizons, up to 2 years.

The results for 10,000 realisations of this experiment appear in Figure \ref{fig:num-exp}. The estimates perform well even for short time horizons. The behaviour of the estimates $\hat{a}$, $\hat{b}$ and $\hat{\sigma}_I$ of the interest process are consistent with known observations from the literature on the maximum likelihood estimation of Cox-Ingersoll-Ross processes. The estimates $\hat{\mu}_P$ and $\hat{\sigma}_P$ of the mean and drift of the price process are very accurate even on short time horizons, as is the estimate $\hat{\tau}$ of the delay parameter.

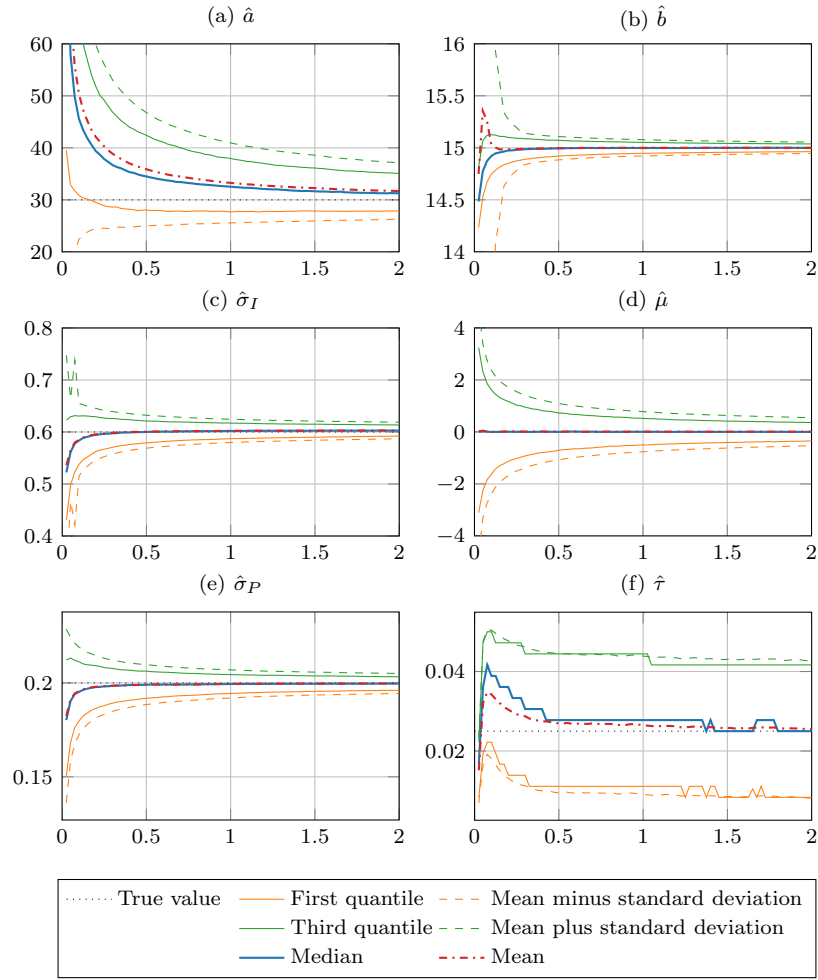
\begin{figure}
\setlength{\tempvSpace}{0.8\vSpace}
\centering
\begin{tikzpicture}
   \pgfplotsset{
      every axis/.prefix style = {
        sixonpage,
        xmin = 0, 
        xmax = 2, 
        xlabel near ticks, 
        font=\footnotesize
      }
   }

   \pgfplotstableread[col sep=comma,]{num-exp-aI.csv} {\data}
	\begin{axis} [name=plotaI, ymin=20, ymax=60]
         \addplot[black,dotted]  coordinates {(0,30) (2,30)};
         \addplot[color2] table [x={T}, y={q1}] {\data};
         \addplot[color2,dashed] table [x={T}, y={meanminus}] {\data};
         \addplot[color3] table [x={T}, y={q3}] {\data};
         \addplot[color3,dashed] table [x={T}, y={meanplus}] {\data};
         \addplot[color1,thick] table [x={T}, y={median}] {\data};
         \addplot[color4,thick,dashdotted] table [x={T}, y={mean}] {\data};
    \end{axis}
	\node[align=center,anchor=north] at ($(plotaI.north)+(0,0.5\vSpace)$) {(a) $\hat{a}$};
	
   \pgfplotstableread[col sep=comma,]{num-exp-bI.csv} {\data}
	\begin{axis} [name=plotbI, anchor = west, at = {($(plotaI.east)+(\hSpace,0)$),}, ymin=14, ymax=16]
         \addplot[black,dotted]  coordinates {(0,10) (2,10)};
         \addplot[color2] table [x={T}, y={q1}] {\data};
         \addplot[color2,dashed] table [x={T}, y={meanminus}] {\data};
         \addplot[color3] table [x={T}, y={q3}] {\data};
         \addplot[color3,dashed] table [x={T}, y={meanplus}] {\data};
         \addplot[color1,thick] table [x={T}, y={median}] {\data};
         \addplot[color4,thick,dashdotted] table [x={T}, y={mean}] {\data};
    \end{axis}
	\node[align=center,anchor=north] at ($(plotbI.north)+(0,0.5\vSpace)$) {(b) $\hat{b}$};

   \pgfplotstableread[col sep=comma,]{num-exp-sigmaI.csv} {\data}
	\begin{axis} [name=plotsigmaI, anchor=north, at = {($(plotaI.south)+(0,-\tempvSpace)$),}, ymin=0.4, ymax=0.8]
         \addplot[black,dotted]  coordinates {(0,0.6) (2,0.6)};
         \addplot[color2] table [x={T}, y={q1}] {\data};
         \addplot[color2,dashed] table [x={T}, y={meanminus}] {\data};
         \addplot[color3] table [x={T}, y={q3}] {\data};
         \addplot[color3,dashed] table [x={T}, y={meanplus}] {\data};
         \addplot[color1,thick] table [x={T}, y={median}] {\data};
         \addplot[color4,thick,dashdotted] table [x={T}, y={mean}] {\data};
    \end{axis}
	\node[align=center,anchor=north] at ($(plotsigmaI.north)+(0,0.5\vSpace)$) {(c) $\hat{\sigma}_I$};

   \pgfplotstableread[col sep=comma,]{num-exp-mu.csv} {\data}
	\begin{axis} [name=plotmu, anchor=west, at = {($(plotsigmaI.east)+(\hSpace,0)$),}, ymin=-4, ymax=4]
         \addplot[black,dotted]  coordinates {(0,0) (2,0)};
         \addplot[color2] table [x={T}, y={q1}] {\data};
         \addplot[color2,dashed] table [x={T}, y={meanminus}] {\data};
         \addplot[color3] table [x={T}, y={q3}] {\data};
         \addplot[color3,dashed] table [x={T}, y={meanplus}] {\data};
         \addplot[color1,thick] table [x={T}, y={median}] {\data};
         \addplot[color4,thick,dashdotted] table [x={T}, y={mean}] {\data};
    \end{axis}
	\node[align=center,anchor=north] at ($(plotmu.north)+(0,0.5\vSpace)$) {(d) $\hat{\mu}$};
	
   \pgfplotstableread[col sep=comma,]{num-exp-sigmaP.csv} {\data}
	\begin{axis} [name=plotsigmaP, anchor=north, at = {($(plotsigmaI.south)+(0,-\tempvSpace)$),}]
         \addplot[black,dotted]  coordinates {(0,0.2) (2,0.2)};
         \addplot[color2] table [x={T}, y={q1}] {\data};
         \addplot[color2,dashed] table [x={T}, y={meanminus}] {\data};
         \addplot[color3] table [x={T}, y={q3}] {\data};
         \addplot[color3,dashed] table [x={T}, y={meanplus}] {\data};
         \addplot[color1,thick] table [x={T}, y={median}] {\data};
         \addplot[color4,thick,dashdotted] table [x={T}, y={mean}] {\data};
    \end{axis}
	\node[align=center,anchor=north] at ($(plotsigmaP.north)+(0,0.5\vSpace)$) {(e) $\hat{\sigma}_P$};
 	
   \pgfplotstableread[col sep=comma,]{num-exp-tau.csv} {\data}
	\begin{axis} [name=plottau, anchor=west, at = {($(plotsigmaP.east)+(\hSpace,0)$),}, legend style={at={(-0.5\hSpace,-\tempvSpace)}, anchor=north, legend columns=3, legend cell align=left}]]

         \addplot[black,dotted]  coordinates {(0,0.025) (2,0.025)};
         \addlegendentry{True value\phantom{q}}

         \addplot[color2] table [x={T}, y={q1}] {\data};
         \addlegendentry{First quantile}
         
         \addplot[color2,dashed] table [x={T}, y={meanminus}] {\data};
         \addlegendentry{Mean minus standard deviation}

      \addlegendimage{empty legend}
      \addlegendentry{}
      
         \addplot[color3] table [x={T}, y={q3}] {\data};
         \addlegendentry{Third quantile}

         \addplot[color3,dashed] table [x={T}, y={meanplus}] {\data};
         \addlegendentry{Mean plus standard deviation}
         
      \addlegendimage{empty legend}
      \addlegendentry{}
      
         \addplot[color1,thick] table [x={T}, y={median}] {\data};
         \addlegendentry{Median\phantom{q}}
         
         \addplot[color4,thick,dashdotted] table [x={T}, y={mean}] {\data};
         \addlegendentry{Mean}

    \end{axis}
	\node[align=center,anchor=north] at ($(plottau.north)+(0,0.5\vSpace)$) {(f) $\hat{\tau}$};
\end{tikzpicture}

\caption{Estimation results for 10,000 realisations and different values of time horizon $T$}
\label{fig:num-exp}
\end{figure}

\section{Risk-neutral measure and option pricing}\label{OptionPricing}

Pricing options in this model is based on techniques developed for two well known models, namely the Black-Scholes-Merton model \citep{black1973pricing} and the Heston model \citep{heston1993closed}. 
We assume throughout this section that $H>\tau$.

\subsection{Change of measure}  \label{RNPSect}

We first determine a probability measure under which the discounted price process becomes a martingale, and with the property that the dynamics of the interest and price process remain tractable. To his end, take two progressively measurable processes $\theta^P=(\theta^P_t)$ and $\theta^I=(\theta^I_t)$ (to be determined below) with respect to the filtration $\mathbb{F}$, and then define the processes $W^{P\ast}=(W^{P\ast}_t)$ and $W^{I\ast}=(W^{I\ast}_t)$ as
\begin{align*}
      W^{P\ast}_t  &=   W^P_t + \int_0^t \theta^P_sds, &
      W^{I\ast}_t  &=   W^I_t + \int_0^t \theta^I_sds
\end{align*}
for all $t$. Let $\mathbb{F}^{P\ast}=(\mathcal{F}^{P\ast}_t)$ and $\mathbb{F}^{I\ast}=(\mathcal{F}^{I\ast}_t)$ be the filtrations generated by $W^{P\ast}$ and $W^{I\ast}$, respectively. The goal is to use the Girsanov theorem to define a new measure $\mathds{Q}\sim\mathds{P}$ such that $W^{P\ast}$ and $W^{I\ast}$ are two independent Brownian motions under~$\mathds{Q}$.

The choice of $\theta^P$ and $\theta^I$ is motivated by requirements on the interest and price processes. In order for the discounted price to be a martingale under $\mathds{Q}$, we need
\begin{equation}\label{RNPPrice}
     dP_t =   r P_t dt + \sigma_P P_t \sqrt{I_{t-\tau}} dW^{P\ast}_t
\end{equation}
together with $P_0=e^x\in(0,\infty)$. Furthermore, we require that the interest process has Cox-Ingersoll-Ross dynamics under $\mathds{Q}$, in other words, it is a strong positive-valued solution to the stochastic differential equation
\begin{equation} \label{eqInterest4}
dI_t = \tilde{a}(\tilde{b} -  I_t) dt + \sigma_I \sqrt{I_t}dW^{I\ast}_t 
\end{equation}
and satisfies the initial condition \eqref{eq:I:past}. Comparing \eqref{eq:I:CIR} with \eqref{eqInterest4}, the change of parameters of $I$ introduces two new parameters---these new parameters will be used in due course to calibrate to option prices. After writing
\begin{align*}
    \tilde{a}&= a + \lambda_a, & \tilde{b}=\frac{ab-\lambda_{ab}}{a+\lambda_a},
\end{align*}
where $\lambda_a,\lambda_{ab} \in \mathds{R} $ and rearranging, we obtain
\begin{align}
    \theta^P_t &= \frac{\mu +\frac{1}{2} \sigma_P^2 I_{t-\tau}  -r}{\sigma_P \sqrt{I_{t-\tau}}} \text{ for all }t\ge0,\label{ThetaP}\\
    \theta^I_t &= \frac{\lambda_{ab}}{\sigma_I\sqrt{I_t}}  + \frac{\lambda_a}{\sigma_I} \sqrt{I_t} \text{ for all }t\ge0.\label{ThetaI}
\end{align}

For any adapted process $\theta$ and Brownian motion $W$, define the stochastic exponential $\mathcal{E}^{\theta,W}$ 
of $\theta$ with respect to $W$ as the process
\[
 \mathcal{E}^{\theta,W}_t = \exp \left\{\int_0^t \theta_sdW_s  -\tfrac{1}{2} \int_{0}^t \theta_s^2 ds \right\} \text{ for all }t.
\]
We now have the following result. 

\begin{theorem} \label{th:risk-neut-prob}
Let the processes $\theta^P$ and $\theta^I$ be defined by \eqref{ThetaP}--\eqref{ThetaI}, and assume that
\begin{align}
    \lambda_a &>-a, & \lambda_{ab} &\le ab - \tfrac{1}{2}\sigma_I^2. \label{Cond12}
\end{align}
Define the process $Z=(Z_t)$ as
\begin{align*}
 Z_t &= \mathcal{E}^{-\theta^P,W^P}_t \mathcal{E}^{-\theta^I,W^I}_t \text{ for all }t,
\end{align*}
and the probability measure $\mathds{Q}\sim\mathds{P}$ on $\mathcal{F}_H$ as
\[
\frac{d\mathds{Q}}{d\mathds{P}} = Z_H.
\]
Then $W^{P\ast}$ and $W^{I\ast}$ are two independent Brownian motions under~$\mathds{Q}$ and the price process $P$ and interest process $I$ are strong solutions to the stochastic differential equations \eqref{RNPPrice}--\eqref{eqInterest4} under $\mathbb{Q}$.
\end{theorem}

It will be shown below that the discounted price process is a martingale under $\mathds{Q}$; for this reason we will refer to $\mathds{Q}$ as the \emph{pricing measure}.

\begin{proof}
By the Girsanov theorem \cite[cf.][Theorem~A.15.1]{Musiela_Rutkowski2005}, it is sufficient to show that $E(Z_H)=1$.

The condition \eqref{Cond12} implies that $\tilde{a}>0$ and $2\tilde{a} \tilde{b}  \geq \sigma_I^2$, which in turn implies that the auxiliary stochastic differential equation with initial condition
\begin{align*}
 d\tilde{I}_t &= \tilde{a}(\tilde{b} -  \tilde{I}_t) dt + \sigma_I \sqrt{\tilde{I}_t}dW_t, & \tilde{I}_0 &= \phi_I(0),
\end{align*}
where $W$ is any Brownian motion, has a strong solution that, almost surely, never exits $(0,\infty)$ \cite[Theorem~6.2.3, Proposition~6.2.4]{Lamberton_Lapeyre1996}. Note furthermore that $\theta^I_t = c(I_t)$ for all $t$, where $b:(0,\infty)\rightarrow\mathbb{R}$ is defined as \[
 c(x) = - \frac{\lambda_{ab}}{\sigma_I\sqrt{x}} - \frac{\lambda_a}{\sigma_I}\sqrt{x} \text{ for all }x\in(0,\infty),
\]
and that
\[
 \tilde{a}(\tilde{b} -  \tilde{I}_t) = a(b-\tilde{I}_t) + \sigma_I\sqrt{\tilde{I}_t}c(\tilde{I}_t).
\]
It then follows from a result by \citet[Corollary 2.2]{mijatovic2012martingale} that the stochastic exponential $\mathcal{E}^{\theta^I,W^I}=\mathcal{E}^{c(I),W^I}$ is a martingale. Conditioning on $\mathcal{F}^I_H$ and using Remark~\ref{RemarkNormal}, we then obtain
\begin{align*}
E(Z_H) &= E\left(\mathcal{E}^{-\theta^I,W^I}_t E\left(\mathcal{E}^{-\theta^P,W^P}_t \middle|\mathcal{F}^I_H\right)\right) = E\left(\mathcal{E}^{-\theta^I,W^I}_t\right) = 1,
\end{align*}
and the proof is complete.
\end{proof}

Comparing \eqref{RNPPrice}--\eqref{eqInterest4} with \eqref{eq:I:CIR} and \eqref{Price1}, the interest and price processes have the same dynamics under the real-world measure $\mathds{P}$ and the pricing measure $\mathds{Q}$, albeit with respect to different Brownian motions, and, for the interest process, changed parameters. The independence of $W^{P\ast}$ and $W^{I\ast}$ under $\mathds{Q}$ means that their natural filtrations $\mathbb{F}^{P\ast}$ and $\mathbb{F}^{I\ast}$ are also independent. It is clear from \eqref{ThetaI} that $\mathcal{F}^{I\ast}_t=\mathcal{F}^I_t$ for all $t$, and combining \eqref{eq:I:past} with \eqref{ThetaP} gives that $\mathcal{F}^{P^\ast}_t = \mathcal{F}^P_t$ when $t<\tau$ and $\mathcal{F}^{P^\ast}_t \subseteq \mathcal{F}^P_t\vee\mathcal{F}^I_{t-\tau}$ when $t\ge\tau$. This means that the general filtration $\mathbb{F}$ defined in \eqref{eq:filtration} and the delayed filtration $\mathbb{F}^{\mathrm{d}}$ defined in \eqref{eq:delayed-filtration} retain the same structure when expressed in terms of $\mathbb{F}^{P\ast}$ and $\mathbb{F}^{I\ast}$, in other words, for all $t$ we have
\begin{align} \label{eq:filtrations-under-Q}
 \mathcal{F}_t &= \mathcal{F}^{P\ast}_t\vee\mathcal{F}^{I\ast}_t, &
 \mathcal{F}^{\mathrm{d}}_t &=     
      \begin{cases}
                  \mathcal{F}^{P\ast}_t  &\text{if } t <\tau,\\
                  \mathcal{F}^{P\ast}_t  \vee \mathcal{F}^{I\ast}_{t-\tau} &\text{if } t\ge\tau.
                  \end{cases}
\end{align}

We conclude this section by establishing the martingale properties of the \emph{discounted price process} $\tilde{P}=(\tilde{P}_t)$. Solving \eqref{RNPPrice} gives
\begin{equation}\label{PriceEquation}
  P_t = e^{x+rt}\mathcal{E}^{\sigma_P\sqrt{I_{\cdot-\tau}},W^{P\ast}}_t \text{ for all }t,
\end{equation}
and thus $\tilde{P}$ can be defined as 
\begin{equation}\label{eq:discounted price}
\tilde{P}_t = e^{-rt}P_t = e^x\mathcal{E}^{\sigma_P\sqrt{I_{\cdot-\tau}},W^{P\ast}}_t \text{ for all } t.
\end{equation}

\begin{proposition}\label{PropositionMart}
The discounted price process $\tilde{P}$ is a martingale under the pricing measure $\mathds{Q}$ of Theorem \ref{th:risk-neut-prob} with respect to both the general filtration~$\mathbb{F}$ and the delayed filtration $\mathbb{F}^{\mathrm{d}}$.
\end{proposition}

\begin{proof}
It is clear from \eqref{eq:discounted price} that $\tilde{P}$ is a positive local martingale under $\mathds{Q}$, hence a supermartingale \cite[ Theorem~7.23]{Klebaner}. For each $t$, conditioning with respect to any $\sigma$-field $\mathcal{G}\supseteq\mathcal{F}^{I\ast}_{t-\tau}$ that is independent of $\mathcal{F}^{P\ast}_t$ and using Remark~\ref{RemarkNormal}, we obtain
\[
E_{\mathds{Q}}(\tilde{P}_t) = e^xE_{\mathds{Q}}\left(E_{\mathds{Q}}\left(\mathcal{E}^{\sigma_P\sqrt{I_{\cdot-\tau}},W^{P\ast}}_t \middle| \mathcal{G} \right)\right) = e^x, 
\]
and it follows that $\tilde{P}$ is a martingale with respect to both $\mathbb{F}$ and $\mathbb{F}^{\mathrm{d}}$.
\end{proof}

\subsection{Characteristic function}

Taking the logarithm in \eqref{PriceEquation}, it follows that the log price satisfies
\begin{equation}\label{eq:PriceEquation-log}
    X_t = x + rt -\tfrac{1}{2}\sigma_P^2 \int_0^t I_{u-\tau}du + \sigma_P \int_0^t \sqrt{I_{u-\tau}} dW^{P\ast}_u \text{ for all }t.
\end{equation}
In this section we derive the conditional characteristic function of the log price given the delayed filtration $\mathbb{F}^{\mathrm{d}}$ and the general filtration $\mathbb{F}$ (see \eqref{eq:filtrations-under-Q}). The conditional characteristic functions are highly tractable, and thus naturally useful for option pricing. The key reason for this is that $X_t$ is lognormally distributed when $t\le\tau$, and when $t>\tau$ it turns out to be the sum of a lognormal random variable plus the value of an independent process resembling the log price in the well known Heston model \citep{heston1993closed}.

We use the pricing measure $\mathds{Q}$ of Theorem \ref{th:risk-neut-prob} throughout this section. Define the characteristic function of a stochastic process $Y=(Y_t)$ given a filtration $\mathbb{G}=(\mathcal{G}_t)$ as
\[
\Phi^{Y|\mathbb{G}}_{t|s}(\lambda) = E_{\mathds{Q}}\left(e^{i\lambda Y_t}\middle|\mathcal{G}_s\right) \text{ for all } \lambda\in\mathds{R}, 0\le s\le t \le H.
\]

The conditional characteristic function of the log price with respect to the delayed filtration $\mathbb{F}^{\mathrm{d}}$ is given as follows. It resembles the characteristic function of the Heston model, with a natural adjustment to take account of the delay $\tau$.

\begin{theorem} \label{th:char-fun-delay}
Fix any $\lambda\in\mathds{R}$. For all $0\le s\le t \le H$, the conditional characteristic function of the log price $X$ with respect to $\mathbb{F}^{\mathrm{d}}$ is given by
\[
 \Phi^{X|\mathbb{F}^{\mathrm{d}}}_{t|s}(\lambda) = 
 \begin{cases}
 1 & \text{if }s=t,\\
  e^{i\lambda X_s + i\lambda r(t-s) - \frac{1}{2}\sigma_P^2(i\lambda + \lambda^2)\int_s^t \phi_I(u-\tau)du} & \text{if } s < t\le \tau, \\
  e^{i\lambda X_s + i\lambda r(t-s) + A(t-s) + I_{s-\tau}B(t-s)} & \text{if } \tau\le s < t, \\
  e^{i\lambda r(t-\tau) + A(t-\tau) + \phi_I(0)B(t-\tau)}\Phi^{X|\mathbb{F}^{\mathrm{d}}}_{\tau|s}(\lambda) &\text{if } s \le \tau < t.
 \end{cases} 
\]
Here the functions $A,B:[0,H]\rightarrow\mathds{R}$ are defined as
\begin{align*}
 A(u) &= \frac{\tilde{a}\tilde{b}}{\sigma_I^2}\left[  (\tilde{a}-d) u - 2\ln \frac{1-ge^{-du}}{1-g}\right], &
 B(u) &= \frac{\tilde{a}-d}{\sigma_I^2} \frac{1-e^{-du}}{1-ge^{-du}}
\end{align*}
for all $u\in[0,H]$, where
\begin{align*}
d &= \sqrt{\tilde{a}^2 + \sigma_P^2 \sigma_I^2 (i\lambda+\lambda^2)} , &
g &= \frac{\tilde{a}  - d}{\tilde{a}  + d}.
\end{align*}
\end{theorem}

\begin{proof}
For $0\le s < t\le \tau$ we have
\begin{equation}  \label{eq:normal-delay}
 \left.X_t - X_s\right|\mathcal{F}^{P\ast}_s \sim \N\left(r(t-s) -\tfrac{1}{2}\sigma_P^2 \int_s^t \phi_I(u-\tau)du, \sigma_P^2 \int_s^t \phi_I(u-\tau) du\right),
\end{equation}
and therefore
\begin{equation} \label{eq:Phi-d-before-tau}
 \Phi^{X|\mathbb{F}^{\mathrm{d}}}_{t|s}(\lambda) = \Phi^{X|\mathbb{F}^{P\ast}}_{t|s}(\lambda) = e^{i\lambda X_s + i\lambda r(t-s) - \frac{1}{2}\sigma_P^2(i\lambda + \lambda^2)\int_s^t \phi_I(u-\tau)du}.
\end{equation}

Now define the stochastic process $Z=(Z_t)_{t\in[0,H-\tau]}$ as 
\begin{equation} \label{ZProcess}
Z_t = x + rt -\tfrac{1}{2}\sigma_P^2 \int_0^t I_udu + \sigma_P \int_0^t \sqrt{I_u} dB^{P\ast}_u \text{ for all }t\in[0,H-\tau],
\end{equation}
where $B^{P\ast}=(B^{P\ast}_u)_{u\in[0,H-\tau]}$ is the Brownian motion given by \[B^{P\ast}_u = W^{P\ast}_{u+\tau} - W^{P\ast}_\tau \text{ for all } u\in[0,H-\tau],\]
and define the stochastic process $V=(V_t)_{t\in[0,H-\tau]}$ as $V_t=\sigma_P^2I_t$ for all $t\in[0,H-\tau]$. It then follows from \eqref{eqInterest4} and \eqref{ZProcess} that $V$ and $Z$ satisfy the stochastic differential equations
\begin{align*}
    dZ_t &= \left(r-\tfrac{1}{2}V_t\right)dt + \sqrt{V_t} dB^{P\ast}_t, \\
    dV_t &= \tilde{a}\left(\tilde{b}\sigma_P^2 -  V_t\right) dt + \sigma_I\sigma_P \sqrt{V_t}dW^{I\ast}_t, 
\end{align*}
with initial conditions $Z_0=x$ and $V_0=\sigma_P^2\phi_I(0)$. This is the volatility and log price, respectively, of a Heston model with zero correlation. The natural filtration in this model is $\mathbb{F}^{\mathrm{H}}=(\mathcal{F}^{\mathrm{H}}_t)_{t\in[0,H-\tau]}$, given by
\begin{equation} \label{eq:FZV-FB-FI}
 \mathcal{F}^{\mathrm{H}}_t = \mathcal{F}^{B\ast}_t \vee \mathcal{F}^{I\ast}_t \text{ for all }t\in[0,H-\tau],
\end{equation}
where $\mathbb{F}^{B\ast}=(\mathcal{F}^{B\ast}_t)_{t\in[0,H-\tau]}$ is the filtration generated by $B^{P\ast}$, and the conditional characteristic function of $Z$ is
\begin{equation} \label{eq:char-fun-Z}
 \Phi^{Z|\mathbb{F}^{\mathrm{H}}}_{t|s}(\lambda) = e^{i\lambda Z_s + i\lambda r(t-s) + A(t-s) + I_sB(t-s)}
\end{equation}
for all $\lambda\in\mathbb{R}$ and $0\le s \le t \le H-\tau$ \cite[cf.][]{Albrecher_Mayer_Schoutens_Tistaert2007}. Observe further from Proposition \ref{prop:FP-FB-Fd} and \eqref{eq:FZV-FB-FI} that
$
 \mathcal{F}^{\mathrm{d}}_{s+\tau} = \mathcal{F}^{P\ast}_\tau \vee \mathcal{F}^{\mathrm{H}}_s.
$
The random variable $Z_t$ is independent of $\mathcal{F}^{P\ast}_\tau$ by the independence of the increments of the Brownian motion $W^{P\ast}$ and the independence between $W^{I\ast}$ and $W^{P\ast}$. It follows that
\begin{align} \label{eq:Phi-Z}
 \Phi^{Z|\mathbb{F}^{\mathrm{d}}}_{t|s+\tau}(\lambda)
 &= \Phi^{Z|\mathbb{F}^{\mathrm{H}}}_{t|s}(\lambda) 
\end{align}
\cite[Section 9.7]{Williams1991}.

Assume now that $t\in(\tau,H]$. We have
\[
X_t = X_\tau - x + Z_{t-\tau} \text{ for all } t\in(\tau,H], 
\]
where $Z_{t-\tau}$ is given by \eqref{ZProcess} (see Remark \ref{ChangeBM}). For every $s\in[\tau,t)$ we have from \eqref{eq:char-fun-Z} and \eqref{eq:Phi-Z} that
\begin{align}
 \Phi^{X|\mathbb{F}^{\mathrm{d}}}_{t|s}(\lambda)
 &= e^{i\lambda (X_\tau - x)}\Phi^{Z|\mathbb{F}^{\mathrm{d}}}_{t-\tau|s}(\lambda) \nonumber\\
 &= e^{i\lambda (X_\tau - x)}\Phi^{Z|\mathbb{F}^{\mathrm{H}}}_{t-\tau|s-\tau}(\lambda) \nonumber\\
 &= e^{i\lambda X_s + i\lambda r(t-s) + A(t-s) + I_{s-\tau}B(t-s)}. \label{eq:Phi-d-after-tau}
\end{align}
 For $s\in[0,\tau]$ we use the tower property of conditional expectation and \eqref{eq:Phi-d-after-tau} to obtain
\begin{align*}
 \Phi^{X|\mathbb{F}^{\mathrm{d}}}_{t|s}(\lambda)
 &= E_{\mathds{Q}}\left(\Phi^{X|\mathbb{F}^{\mathrm{d}}}_{t|\tau}(\lambda)\middle|\mathcal{F}^{\mathrm{d}}_s\right) = e^{i\lambda r(t-\tau) + A(t-\tau) + \phi_I(0)B(t-\tau)}\Phi^{X|\mathbb{F}^{\mathrm{d}}}_{\tau|s}(\lambda).
\end{align*}
This completes the proof.
\end{proof}

The conditional characteristic function of the log price with respect to the general filtration $\mathbb{F}$ defined in \eqref{eq:filtration} is given as follows. It reflects the conditional normality of the log returns on short maturities, whilst for longer maturities the Heston-style dynamics come into play.

\begin{theorem} \label{th:char-fun-general}
Fix any $\lambda\in\mathds{R}$. For all $0\le s\le t \le H$, the conditional characteristic function of the log price $X$ with respect to $\mathbb{F}$ is given by
\[
 \Phi^{X|\mathbb{F}}_{t|s}(\lambda) = 
 \begin{cases}
   1 & \text{if }s=t,\\
   e^{i\lambda X_s + i\lambda r(t-s) - \frac{1}{2}\sigma_P^2(i\lambda + \lambda^2)\int_s^t I_{u-\tau}du} & \text{if } t-\tau\le s<t, \\
  e^{i\lambda r(t-s-\tau) + A(t-s-\tau) + I_sB(t-s-\tau)}\Phi^{X|\mathbb{F}}_{s+\tau|s}(\lambda) & \text{if } s< t-\tau.
 \end{cases} 
\]
where the functions $A,B:[0,H]\rightarrow\mathds{R}$ are defined in Theorem \ref{th:char-fun-delay}.
\end{theorem}

\begin{proof}
If $\max\{t-\tau,0\}\le s<t\le H$, we have by Remark \ref{RemarkNormal} that
\begin{equation} \label{eq:normal-general}
 \left.X_t - X_s\right|\mathcal{F}_s \sim \N\left(r(t-s) -\tfrac{1}{2}\sigma_P^2 \int_s^t I_{u-\tau}du, \sigma_P^2 \int_s^t I_{u-\tau} du\right),
\end{equation}
whence
\[
 \Phi^{X|\mathbb{F}}_{t|s}(\lambda) = e^{i\lambda X_s + i\lambda r(t-s) - \frac{1}{2}\sigma_P^2(i\lambda + \lambda^2)\int_s^t I_{u-\tau}du}.
\]

If $0\le s< t-\tau\le t\le H$, it follows from the tower property and Theorem~\ref{th:char-fun-delay} that
\begin{align*}
 \Phi^{X|\mathbb{F}}_{t|s}(\lambda)
 &= E_{\mathds{Q}}\left(\Phi^{X|\mathbb{F}^{\mathrm{d}}}_{t|s+\tau}(\lambda)\middle|\mathcal{F}_s\right) = e^{i\lambda r(t-s-\tau) + A(t-s-\tau) + I_sB(t-s-\tau)}\Phi^{X|\mathbb{F}}_{s+\tau|s}(\lambda).
\end{align*}
This completes the proof. 
\end{proof}

\subsection{Option pricing}

We now briefly discuss the basic notions of option pricing, as they apply to the options that are currently available on Bitcoin, namely European call options (with put options priced by means of put-call parity), using the pricing measure $\mathds{Q}$. 

The price at time $t\ge0$ under $\mathds{Q}$ of an option with expiration date $T\in[t,H]$ and payoff function $f:(0,\infty)\rightarrow\mathds{R}$ is the discounted expectation
\[
\Pi_{T|t}(f|\mathbb{G}) = e^{-r(T-t)} E_{\mathds{Q}}\left(f(P_T)\middle|\mathcal{G}_t\right)
\]
for all $t\in[0,H]$, where $P$ is the price process and $(\mathbb{G},\mathcal{G}_t)=(\mathbb{F},\mathcal{F}_t)$ or $(\mathbb{G},\mathcal{G}_t)=(\mathbb{F}^{\mathrm{d}},\mathcal{F}^{\mathrm{d}}_t)$.

Careful comparison of the conditional characteristic functions in Theorems~\ref{th:char-fun-delay} and \ref{th:char-fun-general} reveal that $\Phi^{X|\mathbb{F}^{\mathrm{d}}}_{T|0}=\Phi^{X|\mathbb{F}}_{T|0}$ for all $T\in[0,H]$. It follows that 
\[
 \Pi_{T|0}(f|\mathbb{F}) = \Pi_{T|0}(f|\mathbb{F}^{\mathrm{d}}),
\]
and so the price at time $0$ is the same, irrespective of the filtration (general or delay) used. This means that the choice of filtration does not arise when pricing options at time $0$. More generally, both the interest and price processes are adapted to the general filtration $\mathbb{F}$, and $\mathcal{F}_t \supseteq \mathcal{F}^{\mathrm{d}}_t$ for all $t\in[0,T]$, so it seems natural to use it for pricing in practice, rather than the delayed filtration~$\mathbb{F}^{\mathrm{d}}$. It is however worth noticing that $\mathcal{F}_t \subseteq \mathcal{F}^{\mathrm{d}}_{t+\tau}$ whenever $t\le H-\tau$, which means that
\[
 \Pi_{T|t}(f|\mathbb{F}) = e^{-r(\tau-t)} E_{\mathds{Q}}\left(\Pi_{T|t+\tau}(f|\mathbb{F}^{\mathrm{d}})\middle|\mathcal{F}_t\right) \text{ when } t\le T-\tau.
\]
Thus, pricing under the delayed filtration may be useful even if the intention is to price under the general filtration. In any case, the technical considerations for option pricing are similar for either choice of filtration.

Inspection of the characteristic functions in Theorems \ref{th:char-fun-delay} and \ref{th:char-fun-general} reveal that the conditional distribution of the price is lognormal in certain circumstances when $T-t\le \tau$ (see~\eqref{eq:normal-delay} and \eqref{eq:normal-general}). When this is the case, one can calculate option prices directly by using the methodology of the Black-Scholes model with time-varying coefficients \cite[see, for example][Section~8.8]{Wilmott2006}. This leads to explicit option pricing formulae for a large class of options.

Methods based on conditional characteristic function can be used to price options where explicit formulae are not available. The numerical results in Section \ref{sec:RealData} are based on a damped Fourier transform-based method proposed by \cite{carr1999option}. \cite{Rouah2013} and \cite{zhu2009applications} provide more details on this and other methods.

\section{Application to data} \label{sec:RealData}

In this section, we apply the methods of Sections \ref{sec:Inference} and \ref{OptionPricing} to daily data for the two year period from August 2020 to August 2022. We use the historical prices of Bitcoin in US dollars, published on Coinmetrics (\texttt{charts.coinmetrics.io/network-data}), and calculated by their own methodology using the most important exchanges (more details at \texttt{https://docs.coinmetrics.io/methodologies/reference-rates/}). The prices are closing prices in the sense that they are taken at the time when the New York market exchange closes (16:00 EST). The prices appear in Figure~\ref{fig:Bitcoin-prices}.

The period from August 2020 to August 2022 was turbulent for Bitcoin. It initially recovered from low values observed during the COVID-19 pandemic, with all lost value regained, and the price continuing to rise, amid signs of growing acceptance among financial institutions and public listing of the Coinbase exchange, reaching a new all-time high in April 2021 \citep{WSJ2021-02-16}. The price dropped to a lower level soon afterwards, amid a new set of regulations from the Chinese government cracking down on cryptocurrencies, and the announcement by the car manufacturer Tesla that it had suspended vehicle purchases using Bitcoin (introduced in February 2021) due to environmental concerns \citep{CNBC2021-02-08, CNBC2021-05-19, CNBC2021-05-19-Pound, Reuters2021-09-24}.

The price of Bitcoin soon recovered, buoyed by the listing of the first Bitcoin based exchange traded fund and rising demand \citep{Forbes2021-10-20,Forbes2021-09-02,CNET2021-10-18}. After reaching new all-time highs in October and November 2022, Bitcoin prices have been in decline ever since. There are a number of contributing factors for this, such as sensitivity to wider economic conditions, which have been difficult due to concerns about post-COVID economic growth, economic effects from the war in Ukraine and, latterly, difficulties experienced by crypto exchange Binance and the crypto lender Celsius, and the Coinbase exchange announcing redundancies of a large part of its workforce \citep{BBC2022-01-07,WSJ-2022-01-23,TC2022-04-06,Sky2022-06-13,BBC2022-06-14}.

\begin{figure}
\begin{center}
\begin{tikzpicture}
   \pgfplotstableread[col sep=comma,]{Bitcoin-proxies.csv} {\data}
	\begin{axis} [proxystyle, twoonpage]
         \addplot[color1] table [x expr={\coordindex}, y={Bitcoin}] {\data};
    \end{axis}

\end{tikzpicture}
\end{center}
\caption{Bitcoin prices (US dollars), 1 August 2020---1 August 2022}
\label{fig:Bitcoin-prices}
\end{figure}
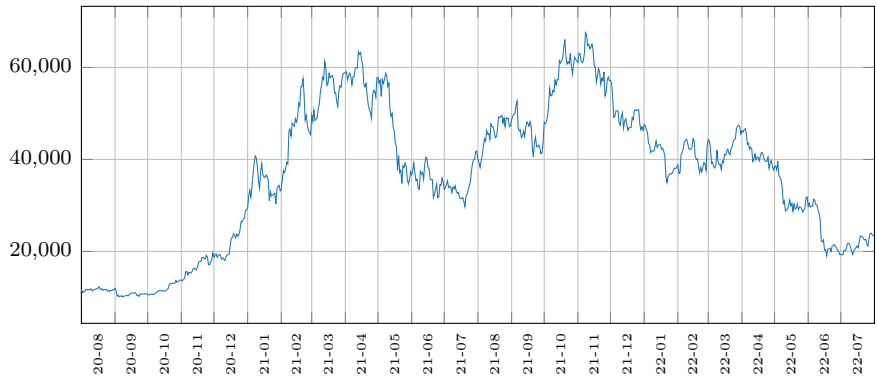

\subsection{Proxies for market attention}
\label{sec:proxies}

\newcommand{\WV}{Wikipedia views} \newcommand{\WVlower}{Wikipedia views}
\newcommand{\LV}{Volume} \newcommand{\LVlower}{volume}
\newcommand{\MR}{Miner revenue} \newcommand{\MRlower}{miner revenue}
\newcommand{\AAdd}{Active addresses} \newcommand{\AAddlower}{active addresses}
\newcommand{\NB}{Number of blocks} \newcommand{\NBlower}{number of blocks}
\newcommand{\BW}{Block weight} \newcommand{\BWlower}{block weight}
\newcommand{\DW}{Deposits and withdrawals} \newcommand{\DWlower}{deposits and withdrawals}
\newcommand{\TF}{Total fees} \newcommand{\TFlower}{total fees}

Whilst market interest itself cannot be observed, it can be approximated by means of proxies that can either be observed or extracted using data mining techniques. A number of such proxies have been studied in the literature, for example, trading volume and number of Google searches \citep{figa2019does,figa2019disentangling}, number of Wikipedia views \citep{kristoufek2015main} or number of tweets \citep{al2021cryptocurrency}. Further examples can be found in the Introduction. In this paper, we consider the proxies listed in Table \ref{table:proxies-info}; these proxies were selected among many possibilities because they appear to be well suited to modeling by a Cox-Ingersoll-Ross process (more details below).

Data for a representative sample consisting of three of the proxies, namely, \WVlower, \LVlower\ and \MRlower, are illustrated in Figure \ref{fig:proxies}. The daily number of \WVlower\ is reasonably consistent over time, and spikes (and periods of higher activity, such as in the early half of 2021) correspond closely to news events involving Bitcoin (summarised above). \LV\ reaches its peak in early 2021, and then gently declines over the rest of the period. Of the three proxies, the graph of \MRlower\ that of Bitcoin the closest, in the sense that it has low initial and final values and two peaks, separated by a deep trough.

\begin{table}
\begin{tabularx}{\linewidth}{>{\raggedright}p{2.25cm}Xd{3}d{3}}
 &  & \multicolumn{2}{c}{\textbf{\boldmath $p$-value} } \\
\textbf{Proxy} & \textbf{Description (Source)} & \multicolumn{1}{c}{\textbf{CIR}} & \multicolumn{1}{c}{\textbf{GBM}} \\
\hline
\WV & Daily number of views in Wikipedia of the keyword ``Altcoin''. (Pageviews Analysis, \texttt{https://pageviews.toolforge.org}) & 0.327	& 0.248 \\
\LV & Natural logarithm of the daily total value in United State dollars of trading volume on major Bitcoin exchanges. (Blockchain, \texttt{www.blockchain.com/charts/trade-volume}) & 0.566 & 0.850 \\
\MR & Natural logarithm of the daily total value in US dollars of all miner revenue (fees plus newly issued Bitcoins). (Coinmetrics) & 0.752	& 0.697 \\
\AAdd & Daily number of active addresses that participate in successful transactions on the blockchain network, either as a sender or as a receiver. (Coinmetrics) & 0.553 & 0.021 \\
\NB & Number of blocks created per day that were included in the main (base) chain. (Coinmetrics) & 0.312 & 0.579 \\
\TF & Natural logarithm of the daily total value in United State dollars of all fees paid to miners, transaction validators, stakers and/or block producers.  (Coinmetrics) & 0.121 & 0.048 \\
\BW & Daily mean weight (a dimensionless measure of a block's ``size'') of all blocks created. (Coinmetrics) & 0.919 & 0.041 \\
\DW & Natural logarithm of the daily average between the total number of Bitcoins withdrawn from exchanges and the total number of Bitcoin sent to exchanges. (Coinmetrics) & 0.255 & 0.025 \\
    \end{tabularx}

        \caption{Market attention proxies with $p$-values of Kolmogorov-Smirnov test applied to Gaussian residuals}
        \label{table:proxies-info}
\end{table}

\begin{figure}
\setlength{\tempvSpace}{1.2\vSpace}
\centering
\begin{tikzpicture}

   \pgfplotsset{
      every axis/.prefix style = {
      proxystyle,
      threeonpage}
    }
   \pgfplotstableread[col sep=comma,]{Bitcoin-proxies.csv} {\data}
   
	\begin{axis} [name=WV, ymin=0]
         \addplot[color1] table [x expr={\coordindex}, y={Wikipedia Views}] {\data};
    \end{axis}
	\node[align=center,anchor=north] at ($(WV.north)+(0,0.5\vSpace)$) {(a) \WV};
	
	\begin{axis} [name=LV, anchor=north, at = {($(WV.south)+(0,-\tempvSpace)$),}]
         \addplot[color2] table [x expr={\coordindex}, y={Log Volume}] {\data};
    \end{axis}
	\node[align=center,anchor=north] at ($(LV.north)+(0,0.5\vSpace)$) {(b) \LV\ ($\ln$ USD)};

	\begin{axis} [name=MR, anchor=north, at = {($(LV.south)+(0,-\tempvSpace)$),}]
         \addplot[color3] table [x expr={\coordindex}, y={Log Miner Revenue (USD)}] {\data};
    \end{axis}
	\node[align=center,anchor=north] at ($(MR.north)+(0,0.5\vSpace)$) {(c) \MR\ ($\ln$ USD)};
\end{tikzpicture}

\caption{\WV, \LVlower\ ($\ln$ USD), \MRlower\ ($\ln$ USD), 1 August 2020---1 August 2022}
\label{fig:proxies}
\end{figure}

In order to verify that the market attention proxies listed in Table \ref{table:proxies-info}, we use a method proposed by \cite{lindstrom2004statistical} involving generalized Gaussian residuals. This method uses the transition distribution function of the diffusion process to generate the Gaussian residuals. This method involves estimating the parameters of the Cox-Ingersoll-Ross process for each proxy, using the techniques of Section~\ref{EST_INT}. The generalised residuals are then computed by applying the conditional cumulative distribution function of the estimated Cox-Ingersoll-Ross process to the data points (see \eqref{CH2:DENCIR} for the accompanying probability density function), followed by the inverse distribution function of the standard normal distribution. It is then possible to test whether the proxy follows a Cox-Ingersoll-Ross process by testing whether the resulting generalised residuals have the properties of a sample from the standard normal distribution. We use the Kolmogorov-Smirnov test \citep[cf.][]{massey1951kolmogorov} for this. The resulting $p$-values appear in Table \ref{table:proxies-info}; the values suggest that the hypothesis or normality (and hence, the hypothesis that the estimated processes fit the data) cannot be rejected at a significance level of 0.1 or lower. In addition, we also verify if the selected market proxies can follow a geometric Brownian motion. The p-values given by the Kolmogorov-Smirnov test are shown in Table \ref{table:proxies-info}. We observe that the hypothesis that market attention follows a geometric Brownian motion cannot be rejected in half of the cases. However, for active addresses, total fees, block weight, as well as deposits and withdrawals, we can reject the null hypothesis with a significance level of 0.05.

\subsection{Market option prices}

Market prices for vanilla European call and put options on Deribit for the first day of each month are published by Tardis ( \texttt{https://docs.tardis.dev/historical-data-details/deribit}); data is taken from 1 August 2021 to 1 August 2022. Tardis publishes intraday data, with bid and ask prices varying throughout the day. For each date and each option, we take the first ask and bid prices of the option in the order book starting at 00:00 UTC.

The underlying asset of the options offered by Deribit is a Bitcoin index based on the prices of Bitcoin in US dollars on a number of exchanges (see  \texttt{https://legacy.deribit.com/pages/docs/options}). The historical values of this index is also published by Tardis. This value is used as the initial value of the underlying for option pricing purposes and to convert option prices from their published values (in Bitcoin) to US dollars. For each option, we take the last value of the Deribit Bitcoin index that appears in the order book before the time at which the option price is taken.

We use mid prices to evaluate our model, but the bid-ask spread in some Bitcoin options can be quite high due to low liquidity. We only take options where the bid and ask prices are within 10\%, in other words,
$$\frac{\text{AskPrice} - \text{BidPrice}}{ \text{AskPrice}} < 0.1.$$
The number of options and their prices vary considerably over time. For the period under consideration, the number of options for each month varies from 106 to 159. The market price for the options chosen also varies considerably due to the variation in Bitcoin prices. Bitcoin option prices on 1 August 2021 and 1 August 2022 are shown in Figure \ref{im:OptionPricesMarket}. Observe that option prices on 1 August 2021 tended to be higher than option prices on  1 August 2022; this is because on the first date, the price of Bitcoin was $\$41501.02$  and on the second date, the Bitcoin price was $\$23307.50$. Apart from that, on 1 August 2021, Bitcoin was in the middle of a bull rally, while one year later, the Bitcoin price was at the end of a decreasing movement; this has implications on the implied volatility as well. The implied volatility is shown in Figure \ref{img:implyVolMarket}. On 1 August 2021, the implied volatility was higher than on  1 August 2022. This is because, in August 2021, Bitcoin was in the middle of an increasing movement, while in August 2022, the price was more stable.

\begin{figure}
\centering
\centering\begin{tikzpicture}
\begin{axis}[
title={Option prices at 01-08-2021},xlabel={Strike},legend columns=5,ymajorgrids=true,width=12cm,height=8cm,scale=1,
xlabel={Strike},  ylabel={Market Price}, scaled x ticks = false ,  x tick label style={font=\tiny,rotate=90,/pgf/number format/fixed},
grid style=dashed,legend style= {at={(0.5,-0.3)},anchor=north} ,]\addplot[color=blue,mark=o,]
coordinates {
(38000, 477.278865)
(41000, 1359.2072025000002)
};
\addplot[color=orange,mark=x,]
coordinates {
(36000, 632.9242575000001)
(37000, 798.9233175)
(38000, 1027.6405425)
(39000, 1348.7168500000002)
(40000, 1701.60291)
(41000, 2147.7548925)
};
\addplot[color=green,mark=square,]
coordinates {
(32000, 435.783915)
(34000, 643.312)
(36000, 996.06024)
(38000, 1504.5200000000002)
(40000, 2251.5111675000003)
};
\addplot[color=red,mark=+,]
coordinates {
(28000, 321.4730225)
(32000, 632.84999)
(34000, 923.4308475)
(35000, 1099.856)
(36000, 1327.37248)
(38000, 1899.5779725)
(40000, 2687.2875225000003)
};
\addplot[color=purple,mark=diamond,]
coordinates {
(26000, 653.7795075)
(28000, 861.3285574999999)
(30000, 1130.3406275)
(32000, 1494.14112)
(34000, 1970.318525)
(36000, 2552.404365)
(38000, 3267.46665)
(40000, 4171.002255)
};
\addplot[color=brown,mark=square*,]
coordinates {
(25000, 944.1984825)
(30000, 1792.964215)
(35000, 3257.947035)
(40000, 5467.955692500001)
};
\addplot[color=pink,mark=triangle,]
coordinates {
(18000, 674.4274875000001)
(20000, 881.8307999999998)
(22000, 1130.9433975)
(24000, 1483.8634525)
(26000, 1877.9885775)
(28000, 2354.9416975)
(30000, 2925.6513)
(32000, 3579.5914875)
(34000, 4337.012295)
(36000, 5167.062495)
(38000, 6111.2445975)
(40000, 7127.384550000001)
};
\addplot[color=yellow,mark=triangle*,]
coordinates {
(25000, 2448.64809)
(30000, 4014.8124975)
(35000, 6100.294199999999)
(40000, 8622.1464525)
};
\addplot[color=black,mark=diamond*,]
coordinates {
(30000, 5083.5491)
(35000, 7356.319897500001)
(40000, 9991.7292825)
};
\legend{T = 0.01,T = 0.03,T = 0.05,T = 0.07,T = 0.15,T = 0.24,T = 0.42,T = 0.65,T = 0.9}
\addplot[color=blue,mark=o,]
coordinates {
(42000, 1421.30376)
(43000, 1058.20818)
(44000, 757.4208075000001)
(45000, 539.552)
};
\addplot[color=orange,mark=x,]
coordinates {
(42000, 2293.0136775)
(43000, 1898.7398325)
(44000, 1566.7197525)
(45000, 1286.57781)
(46000, 1037.74525)
(48000, 653.688)
};
\addplot[color=green,mark=square,]
coordinates {
(42000, 2956.5736125000003)
(44000, 2168.1539824999995)
(46000, 1618.33503)
(48000, 1172.32867)
(50000, 830.0502)
};
\addplot[color=red,mark=+,]
coordinates {
(42000, 3496.25705)
(44000, 2728.53295)
(45000, 2417.7635275)
(46000, 2116.704)
(50000, 1234.6996725)
(55000, 612.2696974999999)
};
\addplot[color=purple,mark=diamond,]
coordinates {
(42000, 5260.6320000000005)
(44000, 4472.056)
(46000, 3797.479665)
(48000, 3247.688)
(50000, 2748.82115)
(56000, 1722.354165)
(60000, 1276.2243225)
(64000, 944.1984825)
};
\addplot[color=brown,mark=square*,]
coordinates {
(45000, 5779.432)
(50000, 4357.76355)
(60000, 2510.992)
(70000, 1514.867895)
(80000, 995.92248)
};
\addplot[color=pink,mark=triangle,]
coordinates {
(42000, 9224.264)
(44000, 8476.089049999999)
(46000, 7822.44086)
(48000, 7200.685485)
(50000, 6649.7099625)
(52000, 6183.425499999999)
(54000, 5749.108685)
(56000, 5345.8068625000005)
(60000, 4627.0939)
(64000, 4035.73885)
(80000, 2459.0237175)
(100000, 1535.61951)
(120000, 1089.4597875000002)
};
\addplot[color=yellow,mark=triangle*,]
coordinates {
(50000, 9150.4413)
(60000, 6982.797307500001)
(70000, 5436.828810000001)
(80000, 4357.9116)
(100000, 2946.6782100000005)
};
\addplot[color=black,mark=diamond*,]
coordinates {
(50000, 11443.23895)
(60000, 9212.63592)
(70000, 7604.61845)
};
\end{axis}
\end{tikzpicture}

\vfill
\centering
\centering\begin{tikzpicture}
\begin{axis}[
title={Option prices at 01-08-2022},xlabel={Strike},legend columns=5,ymajorgrids=true,width=12cm,height=8cm,scale=1,
xlabel={Strike},  ylabel={Market Price}, scaled x ticks = false ,  x tick label style={font=\tiny,rotate=90,/pgf/number format/fixed},
grid style=dashed,legend style= {at={(0.5,-0.3)},anchor=north} ,]\addplot[color=blue,mark=o,]
coordinates {
(21500, 192.2980125)
(22000, 291.372625)
(22500, 431.231485)
(23000, 623.4681350000001)
};
\addplot[color=orange,mark=x,]
coordinates {
(19500, 180.46518250000003)
(20000, 244.753005)
(20500, 314.682435)
(21000, 413.7491275)
(21500, 536.12563)
(22000, 681.8119425)
(22500, 862.4526099999999)
(23000, 1078.0343125)
};
\addplot[color=green,mark=square,]
coordinates {
(19000, 279.62748)
(20500, 530.2763375)
(21000, 658.4942225)
(21500, 798.3609925000001)
(22000, 973.1845675)
(22500, 1165.471)
(23000, 1392.7037875)
};
\addplot[color=red,mark=+,]
coordinates {
(19000, 407.80390000000006)
(20000, 594.375675)
(21000, 850.7938299999998)
(22000, 1206.2329874999998)
(23000, 1637.5141525000004)
};
\addplot[color=purple,mark=diamond,]
coordinates {
(16000, 448.7138425)
(17000, 588.5727025000001)
(18000, 780.869255)
(19000, 1013.964555)
(20000, 1305.3493600000002)
(21000, 1643.3416050000003)
(22000, 2051.23864)
(23000, 2523.2566225)
};
\addplot[color=brown,mark=square*,]
coordinates {
(14000, 419.5765800000001)
(16000, 705.1175175)
(17000, 897.427685)
(18000, 1130.518995)
(19000, 1416.0709575)
(20000, 1742.3978324999998)
(21000, 2127.0073875)
(22000, 2564.0791000000004)
(23000, 3047.7393525)
};
\addplot[color=pink,mark=triangle,]
coordinates {
(14000, 815.8335500000001)
(15000, 1031.4590925)
(16000, 1264.5571925000002)
(17000, 1538.44746)
(18000, 1847.3024425)
(19000, 2196.9495925)
(20000, 2575.7340050000003)
(21000, 3001.1380375000003)
(22000, 3473.14083)
(23000, 3985.953569999999)
};
\addplot[color=yellow,mark=triangle*,]
coordinates {
(10000, 559.4320799999999)
(13000, 1112.72207)
(15000, 1636.0993175)
(16000, 1934.7142300000005)
(17000, 2278.5339275)
(18000, 2663.1457925000004)
(19000, 3071.0674675000005)
(20000, 3519.7813100000003)
(21000, 3991.8049625)
(22000, 4498.79333)
(23000, 5034.91896)
};
\addplot[color=black,mark=diamond*,]
coordinates {
(10000, 798.3514025000001)
(13000, 1497.6373025)
(15000, 2113.00485)
(16000, 2488.1749025)
(17000, 2875.165215)
(18000, 3298.1428450000003)
(19000, 3741.003015)
(20000, 4213.2481575)
(21000, 4737.7188825)
(22000, 5219.712960000001)
(23000, 5779.4312925)
};
\legend{T = 0.01,T = 0.03,T = 0.05,T = 0.07,T = 0.16,T = 0.24,T = 0.41,T = 0.66,T = 0.91}
\addplot[color=blue,mark=o,]
coordinates {
(23500, 675.9844899999999)
(24000, 477.8511050000001)
(24500, 332.10807750000004)
(25000, 227.2612875)
};
\addplot[color=orange,mark=x,]
coordinates {
(23500, 1136.3064375)
(24000, 920.737495)
(24500, 745.9139200000001)
(25000, 594.400155)
(25500, 466.1592)
(26000, 361.287175)
};
\addplot[color=green,mark=square,]
coordinates {
(23500, 1462.6730075)
(24000, 1241.2473825)
(24500, 1048.94145)
(25000, 879.9453275000001)
(25500, 740.0864675)
(26000, 611.8825125)
(27000, 419.5765800000001)
(28000, 279.71436)
};
\addplot[color=red,mark=+,]
coordinates {
(24000, 1503.226035)
(25000, 1130.525785)
(26000, 833.3257075000001)
(27000, 611.8825125)
(28000, 442.88639)
(30000, 233.05830000000003)
};
\addplot[color=purple,mark=diamond,]
coordinates {
(24000, 2453.3575025)
(25000, 2051.25096)
(26000, 1701.6059099999998)
(28000, 1171.3109175)
(30000, 798.3561975)
(32000, 553.598725)
(34000, 390.432785)
(35000, 332.159235)
};
\addplot[color=brown,mark=square*,]
coordinates {
(24000, 3036.1027525000004)
(25000, 2628.1810775000004)
(26000, 2266.8790225)
(27000, 1952.1848625)
(28000, 1678.28616)
(29000, 1445.19086)
(30000, 1247.059855)
(32000, 926.5538175)
(34000, 705.1132825)
(36000, 530.2918075)
(38000, 413.7441575)
(40000, 326.33538)
};
\addplot[color=pink,mark=triangle,]
coordinates {
(24000, 4079.16775)
(25000, 3694.560505)
(26000, 3327.4354075)
(27000, 3006.96549)
(28000, 2709.7328625)
(30000, 2202.763815)
(32000, 1800.6611925)
(34000, 1474.3366274999998)
(35000, 1334.4786075)
(36000, 1212.11012)
(38000, 1013.964555)
(40000, 856.6252274999999)
(45000, 564.717995)
(50000, 402.0918075)
(60000, 233.0981)
};
\addplot[color=yellow,mark=triangle*,]
coordinates {
(24000, 5266.496080000001)
(25000, 4887.6553275)
(26000, 4527.930592500001)
(27000, 4192.1712)
(28000, 3889.40328)
(29000, 3610.35455)
(30000, 3356.57232)
(32000, 2906.0524975)
(34000, 2523.2869325)
(35000, 2342.621835)
(36000, 2191.09582)
(38000, 1905.5540775)
(40000, 1678.29624)
(45000, 1229.5777075)
(50000, 914.8990525)
(60000, 565.2561025)
};
\addplot[color=black,mark=diamond*,]
coordinates {
(24000, 6194.215272500001)
(25000, 5827.1075)
(26000, 5477.481049999999)
(27000, 5150.588430000001)
(28000, 4848.15344)
(29000, 4573.0744125)
(30000, 4312.31485)
(40000, 2450.2926225)
(45000, 1946.18293)
(50000, 1538.3002800000002)
};
\end{axis}
\end{tikzpicture}
\caption{Bitcoin Option Prices on 1 August 2021 and 1 August 2022}
\label{im:OptionPricesMarket}
\end{figure}

\begin{figure}[ht]

\centering\begin{tikzpicture}
\begin{axis}[
title={Implied volatility 1 August 2021},xlabel={Strike},legend columns=5,ymajorgrids=true,width=12cm,height=8cm,
 scaled x ticks = false ,  x tick label style={font=\tiny,rotate=90,/pgf/number format/fixed},
grid style=dashed,legend style= {at={(0.5,-0.3)},anchor=north} ,]\addplot[color=blue,mark=o,]
coordinates {
(38000, 0.9051967295805536)
(41000, 0.8299292266566862)
};
\addplot[color=orange,mark=x,]
coordinates {
(36000, 0.8826883733840732)
(37000, 0.8506443188906259)
(38000, 0.8309866367090696)
(39000, 0.8231285269913635)
(40000, 0.8074188911057723)
(41000, 0.8018722254135379)
};
\addplot[color=green,mark=square,]
coordinates {
(32000, 0.9463313528987498)
(34000, 0.8842686678627588)
(36000, 0.8451622642501008)
(38000, 0.8112278446152903)
(40000, 0.7943126453873789)
};
\addplot[color=red,mark=+,]
coordinates {
(28000, 1.0314902484892705)
(32000, 0.9062732034221201)
(34000, 0.8637926954759865)
(35000, 0.8409234894793729)
(36000, 0.8236705568540623)
(38000, 0.7984287623845057)
(40000, 0.7817776890515882)
};
\addplot[color=purple,mark=diamond,]
coordinates {
(26000, 0.9900896449560083)
(28000, 0.9421267419068642)
(30000, 0.8980797186240721)
(32000, 0.8652296003795373)
(34000, 0.8376253454902068)
(36000, 0.8138157178743617)
(38000, 0.7927448469075641)
(40000, 0.7843394048648387)
};
\addplot[color=brown,mark=square*,]
coordinates {
(25000, 0.9161850349567833)
(30000, 0.8387453991334834)
(35000, 0.7969964733962687)
(40000, 0.7758985332798554)
};
\addplot[color=pink,mark=triangle,]
coordinates {
(18000, 0.9609446536959699)
(20000, 0.9216077660710937)
(22000, 0.8859198459661664)
(24000, 0.8650067188657738)
(26000, 0.841092058289875)
(28000, 0.8215324872121381)
(30000, 0.8063496886468768)
(32000, 0.792713736372744)
(34000, 0.782345242843595)
(36000, 0.7713383786968218)
(38000, 0.7641537001043176)
(40000, 0.7561758946749245)
};
\addplot[color=yellow,mark=triangle*,]
coordinates {
(25000, 0.7972821571610139)
(30000, 0.7651755705071421)
(35000, 0.744424218834836)
(40000, 0.7258119233437592)
};
\addplot[color=black,mark=diamond*,]
coordinates {
(30000, 0.7457075515609504)
(35000, 0.7277151690372914)
(40000, 0.7102900423986632)
};
\legend{T = 0.01,T = 0.03,T = 0.05,T = 0.07,T = 0.15,T = 0.24,T = 0.42,T = 0.65,T = 0.9}
\addplot[color=blue,mark=o,]
coordinates {
(42000, 0.8520626418955431)
(43000, 0.8646675959524024)
(44000, 0.8654585805860611)
(45000, 0.873787414843718)
};
\addplot[color=orange,mark=x,]
coordinates {
(42000, 0.8398582482364406)
(43000, 0.8446134659689115)
(44000, 0.8511823453304744)
(45000, 0.8579728384477893)
(46000, 0.8584154239602556)
(48000, 0.8583729630127236)
};
\addplot[color=green,mark=square,]
coordinates {
(42000, 0.8440595116275863)
(44000, 0.8416303024240785)
(46000, 0.8578089576934692)
(48000, 0.8634141049320533)
(50000, 0.86497658352186)
};
\addplot[color=red,mark=+,]
coordinates {
(42000, 0.8434917691940338)
(44000, 0.8480165750220934)
(45000, 0.8533051334964107)
(46000, 0.8539104993375076)
(50000, 0.862292299799338)
(55000, 0.8731429690904655)
};
\addplot[color=purple,mark=diamond,]
coordinates {
(42000, 0.863220833700318)
(44000, 0.8627366681532066)
(46000, 0.8640958824986924)
(48000, 0.8704321724157515)
(50000, 0.8726290866739093)
(56000, 0.890705969804094)
(60000, 0.904202022586728)
(64000, 0.9146921252427286)
};
\addplot[color=brown,mark=square*,]
coordinates {
(45000, 0.8753011976052203)
(50000, 0.8848406005294578)
(60000, 0.9024985884101506)
(70000, 0.924091021266506)
(80000, 0.955364847246577)
};
\addplot[color=pink,mark=triangle,]
coordinates {
(42000, 0.893266682493168)
(44000, 0.8915251912731091)
(46000, 0.8929359925330843)
(48000, 0.8920397138936201)
(50000, 0.8936159179631594)
(52000, 0.8983665795130085)
(54000, 0.9019628581767833)
(56000, 0.9049284267593738)
(60000, 0.9113819482456126)
(64000, 0.9182727225588089)
(80000, 0.9458915980134119)
(100000, 0.9904231943979989)
(120000, 1.0374517896287476)
};
\addplot[color=yellow,mark=triangle*,]
coordinates {
(50000, 0.9052142408576075)
(60000, 0.9132317234438172)
(70000, 0.9213404510019596)
(80000, 0.9330932798337298)
(100000, 0.9542232734718823)
};
\addplot[color=black,mark=diamond*,]
coordinates {
(50000, 0.9174084863634939)
(60000, 0.9186231295346816)
(70000, 0.925501325728224)
};
\end{axis}
\end{tikzpicture}

\vfill
\centering\begin{tikzpicture}
\begin{axis}[
title={Implied volatility 1 August 2022},xlabel={Strike},legend columns=5,ymajorgrids=true,width=12cm,height=8cm, scaled x ticks = false ,  x tick label style={font=\tiny,rotate=90,/pgf/number format/fixed},
grid style=dashed,legend style= {at={(0.5,-0.3)},anchor=north} ,]

\addplot[color=blue,mark=o]
coordinates {
(21500, 0.8422363906879659)
(22000, 0.8215892627387488)
(22500, 0.8045172223280058)
(23000, 0.7937622490833398)
};
\addplot[color=orange,mark=x]
coordinates {
(19500, 0.8573162827441444)
(20000, 0.8480605474634514)
(20500, 0.8239536964446731)
(21000, 0.8105760926377944)
(21500, 0.7975848080521188)
(22000, 0.7829202662749833)
(22500, 0.7728252333634525)
(23000, 0.7652787882897447)
};
\addplot[color=green,mark=square]
coordinates {
(19000, 0.8477702783918482)
(20500, 0.7984618093674135)
(21000, 0.7897282259792106)
(21500, 0.7755680264990309)
(22000, 0.7685815808058909)
(22500, 0.7581050318015351)
(23000, 0.7526991619336205)
};
\addplot[color=red,mark=+]
coordinates {
(19000, 0.8177163077709757)
(20000, 0.7905881299714607)
(21000, 0.7676709160226527)
(22000, 0.7560253582147972)
(23000, 0.7408214626534423)
};
\addplot[color=purple,mark=diamond]
coordinates {
(16000, 0.8435464120505299)
(17000, 0.8136713567299573)
(18000, 0.7933043794872255)
(19000, 0.7729762496241359)
(20000, 0.7570999649387721)
(21000, 0.7398548456652531)
(22000, 0.7273099335685517)
(23000, 0.7168114040318264)
};
\addplot[color=brown,mark=square*]
coordinates {
(14000, 0.8625610882326581)
(16000, 0.8122777018368968)
(17000, 0.79057013369691)
(18000, 0.7714335165136886)
(19000, 0.7567089941652367)
(20000, 0.7415342445001851)
(21000, 0.7299051058328535)
(22000, 0.7196179893173661)
(23000, 0.7091287544607642)
};
\addplot[color=pink,mark=triangle]
coordinates {
(14000, 0.8158848897704288)
(15000, 0.8020785052122629)
(16000, 0.784586291309405)
(17000, 0.7700806342670057)
(18000, 0.7562534179720081)
(19000, 0.744030223830669)
(20000, 0.7307318312529443)
(21000, 0.7197366191677916)
(22000, 0.7107171703593491)
(23000, 0.7025624065384531)
};
\addplot[color=yellow,mark=triangle*]
coordinates {
(10000, 0.8572770375023334)
(13000, 0.800605515685068)
(15000, 0.7727056347555836)
(16000, 0.7591689364125911)
(17000, 0.7485357856127438)
(18000, 0.7399943209967657)
(19000, 0.7303901669358074)
(20000, 0.722668453390613)
(21000, 0.7140865424087346)
(22000, 0.7065032719452026)
(23000, 0.6991238588106491)
};
\addplot[color=black,mark=diamond*]
coordinates {
(10000, 0.8156882830903345)
(13000, 0.7705738210896664)
(15000, 0.7464354489857975)
(16000, 0.7405575843924174)
(17000, 0.7304382592544454)
(18000, 0.7236752411839827)
(19000, 0.715157895911961)
(20000, 0.7073231463089232)
(21000, 0.7029830663810134)
(22000, 0.6899832684238708)
(23000, 0.6831686911766021)
};
\legend{T = 0.01,T = 0.03,T = 0.05,T = 0.07,T = 0.16,T = 0.24,T = 0.41,T = 0.66,T = 0.91}

\addplot[color=blue,mark=o]
coordinates {
(23500, 0.7852491347921841)
(24000, 0.7842707978268948)
(24500, 0.791462501891705)
(25000, 0.7993687914080617)
};
\addplot[color=orange,mark=x]
coordinates {
(23500, 0.7590904623541735)
(24000, 0.7546100951251172)
(24500, 0.7569604740744353)
(25000, 0.7562894935302488)
(25500, 0.7543055393168638)
(26000, 0.7519733505313464)
};
\addplot[color=green,mark=square]
coordinates {
(23500, 0.7513223639772638)
(24000, 0.7463220315947602)
(24500, 0.7437036245634753)
(25000, 0.7411611545944154)
(25500, 0.7426242595920459)
(26000, 0.7397720484384733)
(27000, 0.7413870301900801)
(28000, 0.7408863509522325)
};
\addplot[color=red,mark=+]
coordinates {
(24000, 0.74199536202183)
(25000, 0.7367241515607063)
(26000, 0.732242175060305)
(27000, 0.7328104724858422)
(28000, 0.7331111910346308)
(30000, 0.7421375724081637)
};
\addplot[color=purple,mark=diamond]
coordinates {
(24000, 0.7305631276671725)
(25000, 0.7236545099147234)
(26000, 0.7170943965430275)
(28000, 0.7130556228415355)
(30000, 0.7118857013255246)
(32000, 0.7177487800521851)
(34000, 0.7269132558864652)
(35000, 0.733404741290333)
};
\addplot[color=brown,mark=square*]
coordinates {
(24000, 0.7312343713310957)
(25000, 0.7240882981910388)
(26000, 0.7179716450236965)
(27000, 0.7134758365962973)
(28000, 0.7100052546513508)
(29000, 0.7083974694542305)
(30000, 0.7082375890996797)
(32000, 0.7084563980411048)
(34000, 0.7150571630670025)
(36000, 0.7182143969439176)
(38000, 0.7276830251559714)
(40000, 0.7373104051206517)
};
\addplot[color=pink,mark=triangle]
coordinates {
(24000, 0.7343324786055336)
(25000, 0.7312791457031846)
(26000, 0.7259249180438972)
(27000, 0.7234333764021471)
(28000, 0.7202600948679736)
(30000, 0.7159757032589412)
(32000, 0.7148989134973459)
(34000, 0.7145611994879552)
(35000, 0.7144830138990322)
(36000, 0.7152868604331649)
(38000, 0.7198465555991337)
(40000, 0.725557540557616)
(45000, 0.7371331814989668)
(50000, 0.7545812944612748)
(60000, 0.7938085331549943)
};
\addplot[color=yellow,mark=triangle*]
coordinates {
(24000, 0.7410762267568322)
(25000, 0.7373012202365776)
(26000, 0.732311880298567)
(27000, 0.7294773447549253)
(28000, 0.7267475384193917)
(29000, 0.7243293865027572)
(30000, 0.7219430847417169)
(32000, 0.720817922918072)
(34000, 0.7190454246673885)
(35000, 0.717232626106172)
(36000, 0.7178545303274384)
(38000, 0.7168805856817958)
(40000, 0.7190538959872242)
(45000, 0.7231882499283873)
(50000, 0.727315624404546)
(60000, 0.7460221679813923)
};
\addplot[color=black,mark=diamond*]
coordinates {
(24000, 0.7401086594439917)
(25000, 0.7364739774612933)
(26000, 0.7325987476447913)
(27000, 0.7293618203059209)
(28000, 0.7265606956470119)
(29000, 0.7254630878557745)
(30000, 0.7233998792779742)
(40000, 0.714664829157678)
(45000, 0.7206914362273338)
(50000, 0.7227761286359821)
};
\end{axis}
\end{tikzpicture}
\caption{Implied volatility on 1 August 2021 and 1 August 2022}
\label{img:implyVolMarket}
\end{figure}

\subsection{Estimation and calibration}

For each choice of the proxy and for each date for which option data is available, the parameters $a$, $b$, $\sigma_I$, $\mu$, $\sigma_P$ and $\tau$ estimated from historical market data and using the techniques of Section \ref{sec:Inference}. A temporal window of one year is used; for example, to price options on 1 August 2021 (the case labeled as ``21-08'' in all graphs below), parameters are estimated using data from 1 August 2020 to 1 August 2021. Data is taken daily, i.e.~$\Delta=\nicefrac{1}{365}$.

As we have seen in Section \ref{RNPSect}, we can pick any value for the change-of-measure parameters $\lambda_{ab}$ and $\lambda_a$ defined in \eqref{ThetaI}, as long as \eqref{Cond12} is satisfied. This gives a certain flexibility to our model. The pricing measure corresponding to $\lambda_{ab}=0$ and $\lambda_a=0$ (and hence $\tilde{a}=a$ and $\tilde{b}=b$) has the property that the dynamics of the interest process is the same under the physical measure $\mathds{P}$ and the pricing measure $\mathds{Q}$. Below, we will refer to the model with this choice of parameters as the uncalibrated version of the model.

In general, it makes sense to choose values for $\lambda_{ab}$ and $\lambda_a$ (and hence $\tilde{a}$ and~$\tilde{b}$) that are optimal in some sense. For such choices, the drift term of the interest process (including the speed of mean reversion and the long term mean) will be different under the real-world and risk-neutral measures, but the diffusion part will be the same under both measures. For each date, we choose these parameters so as to minimise the root mean square error (RMSE) when using the model to approximate option prices, defined as 
$$\text{RMSE}=\sqrt{\frac{1}{n}\sum_{j=1}^n \left(\text{MidPrice}_j - \text{ModelPrice}_j\right)^2},$$
where $n$ is the number of options and $\text{MidPrice}_j=\tfrac{1}{2}(\text{AskPrice}_j + \text{BidPrice}_j)$. We refer to the model with this choice of parameters below as the calibrated version of the model. 

We take $r=0$ when pricing options, as real interest rates were very low during the period of consideration. The computer implementation of the estimation and calibration procedures was done in Python, and in particular heavy use was made of the sequential least squares optimisation procedure in SciPy \citep{2020SciPy-NMeth}.

The estimated parameters for \WVlower, \LVlower\ and \MRlower, are illustrated in Figure \ref{fig:InferenceResults}. The estimated parameters of the uncalibrated model are $\hat{a}$, $\hat{b}$, $\hat{\sigma}_I$, $\hat{\mu}$, $\hat{\sigma}_P$ and $\hat{\tau}$, and the calibrated model corresponds to the estimated parameters $\hat{\tilde{a}}$, $\hat{\tilde{b}}$, $\hat{\sigma}_I$, $\hat{\mu}$, $\hat{\sigma}_P$ and $\hat{\tau}$. The estimates $\hat{\tilde{a}}$ and $\hat{\tilde{b}}$ are clearly very different from $\hat{a}$ and $\hat{b}$ in most cases, which suggests that calibration considerably improves the fit to market option prices; this will also be demonstrated in the next section. The calibrated estimates $\hat{\tilde{a}}$ and $\hat{\tilde{b}}$ appear to fluctuate considerably more than their real-world counterparts $\hat{a}$ and $\hat{b}$, which suggests that the optimal risk-neutral measure varies considerably over time.

\begin{figure}
\setlength{\tempvSpace}{1.2\vSpace}
\begin{center}
\begin{tikzpicture}
   \pgfplotsset{
      every axis/.prefix style = {
      every axis plot/.append style={dotted, mark options={solid, thick}},
      eightonpage,
      x tick label style={font=\tiny, rotate=90},
      xtick = {0, 1, 2, 3, 4, 5, 6, 7, 8, 9, 10, 11, 12},
      xticklabels={21-08, 21-09, 21-10, 21-11, 21-12, 22-01, 22-02, 22-03, 22-04, 22-05, 22-06, 22-07, 22-08},
      xmin = 0,
      xmax = 12}
    }

   \pgfplotstableread[col sep=comma,]{inference-a.csv} {\data}
	\begin{axis} [name=plotaI]
	    \addplot+ table [x expr={\coordindex}, y={Wikipedia Views}] {\data};
        \addplot+ table [x expr={\coordindex}, y={Log Volume}] {\data};
        \addplot+ table [x expr={\coordindex}, y={Log Miner Revenue (USD)}] {\data};
    \end{axis}
	\node[align=center,anchor=north] at ($(plotaI.north)+(0,0.5\vSpace)$) {(a) $\hat{a}$};
	
   \pgfplotstableread[col sep=comma,]{inference-b.csv} {\data}
	\begin{axis} [name=plotbI, anchor = west, at = {($(plotaI.east)+(\hSpace,0)$)}]
	    \addplot+ table [x expr={\coordindex}, y={Wikipedia Views}] {\data};
        \addplot+ table [x expr={\coordindex}, y={Log Volume}] {\data};
        \addplot+ table [x expr={\coordindex}, y={Log Miner Revenue (USD)}] {\data};
    \end{axis}
	\node[align=center,anchor=north] at ($(plotbI.north)+(0,0.5\vSpace)$) {(b) $\hat{b}$};

   \pgfplotstableread[col sep=comma,]{inference-atilde.csv} {\data}
	\begin{axis} [name=plotaItilde, anchor=north, at = {($(plotaI.south)+(0,-\tempvSpace)$),}]
	    \addplot+ table [x expr={\coordindex}, y={Wikipedia Views}] {\data};
        \addplot+ table [x expr={\coordindex}, y={Log Volume}] {\data};
        \addplot+ table [x expr={\coordindex}, y={Log Miner Revenue (USD)}] {\data};
    \end{axis}
	\node[align=center,anchor=north] at ($(plotaItilde.north)+(0,0.5\vSpace)$) {(c) $\hat{\tilde{a}}$};
	
   \pgfplotstableread[col sep=comma,]{inference-btilde.csv} {\data}
	\begin{axis} [name=plotbItilde, anchor = west, at = {($(plotaItilde.east)+(\hSpace,0)$)}]
	    \addplot+ table [x expr={\coordindex}, y={Wikipedia Views}] {\data};
        \addplot+ table [x expr={\coordindex}, y={Log Volume}] {\data};
        \addplot+ table [x expr={\coordindex}, y={Log Miner Revenue (USD)}] {\data};
    \end{axis}
	\node[align=center,anchor=north] at ($(plotbItilde.north)+(0,0.5\vSpace)$) {(d) $\hat{\tilde{b}}$};

   \pgfplotstableread[col sep=comma,]{inference-sigmaI.csv} {\data}
	\begin{axis} [name=plotsigmaI, anchor=north, at = {($(plotaItilde.south)+(0,-\tempvSpace)$),}]
	    \addplot+ table [x expr={\coordindex}, y={Wikipedia Views}] {\data};
        \addplot+ table [x expr={\coordindex}, y={Log Volume}] {\data};
        \addplot+ table [x expr={\coordindex}, y={Log Miner Revenue (USD)}] {\data};
    \end{axis}
	\node[align=center,anchor=north] at ($(plotsigmaI.north)+(0,0.5\vSpace)$) {(e) $\hat{\sigma}_I$};

   \pgfplotstableread[col sep=comma,]{inference-mu.csv} {\data}
	\begin{axis} [name=plotmu, anchor=west, at = {($(plotsigmaI.east)+(\hSpace,0)$),}]
	    \addplot+ table [x expr={\coordindex}, y={Wikipedia Views}] {\data};
        \addplot+ table [x expr={\coordindex}, y={Log Volume}] {\data};
        \addplot+ table [x expr={\coordindex}, y={Log Miner Revenue (USD)}] {\data};
    \end{axis}
	\node[align=center,anchor=north] at ($(plotmu.north)+(0,0.5\vSpace)$) {(f) $\hat{\mu}$};
	
   \pgfplotstableread[col sep=comma,]{inference-sigmaP.csv} {\data}
	\begin{axis} [name=plotsigmaP, anchor=north, at = {($(plotsigmaI.south)+(0,-\tempvSpace)$),}]
	    \addplot+ table [x expr={\coordindex}, y={Wikipedia Views}] {\data};
        \addplot+ table [x expr={\coordindex}, y={Log Volume}] {\data};
        \addplot+ table [x expr={\coordindex}, y={Log Miner Revenue (USD)}] {\data};
    \end{axis}
	\node[align=center,anchor=north] at ($(plotsigmaP.north)+(0,0.5\vSpace)$) {(g) $\hat{\sigma}_P$};
 	
   \pgfplotstableread[col sep=comma,]{inference-tau.csv} {\data}
	\begin{axis} [name=plottau, anchor=west, at = {($(plotsigmaP.east)+(\hSpace,0)$),}, legend style={at={(-0.5\hSpace,-\vSpace)}, anchor=north, legend columns=3, legend cell align=left}, scaled y ticks=false, yticklabel style = {/pgf/number format/fixed},]

	    \addplot+ table [x expr={\coordindex}, y={Wikipedia Views}] {\data};
	    \addlegendentry{\WV}

        \addplot+ table [x expr={\coordindex}, y={Log Volume}] {\data};
	    \addlegendentry{\LV}
        
        \addplot+ table [x expr={\coordindex}, y={Log Miner Revenue (USD)}] {\data};
	    \addlegendentry{\MR}

    \end{axis}
	\node[align=center,anchor=north] at ($(plottau.north)+(0,0.5\vSpace)$) {(h) $\hat{\tau} / \Delta $};

\end{tikzpicture}
\end{center}
\caption{Estimated parameters using \WV, \LVlower, \MRlower}
\label{fig:InferenceResults}
\end{figure}

Looking more closely at the individual proxies, many of the estimates for \LVlower\ and \MRlower\ follow broadly the same pattern: broadly constant long term mean and volatility of the interest process under the real-world measure, and the drift of the price process declining steadily over the period, with its volatility staying broadly constant. The estimates for \WVlower\ share some of these features, but not all; most notably its long term mean declines considerably over time. It is also interesting to note that the estimated delay is different for the proxies: for \MRlower\ it is around two weeks, whereas \LVlower appears to have a much more immediate effect. \WV\ affects Bitcoin prices with a delay of just over a week for the most part, with the delay extending up to two weeks in the early days of the Russia-Ukraine war, and then reducing to zero in the period leading up to August 2022.

\subsection{Model comparison}

The effectiveness of our model can be assessed by comparing it with a number of other models relevant to this setting. In each case, we will do this by dividing the relative root mean square error of our model with the relative root mean square error of the comparator model for each of the dates on which option prices are available. In each case, the comparator model is estimated using the same data (historical and, where relevant, option price data) as our model. When the presented value is lower than $1$, our model has a smaller root mean square error (i.e. performs better) than the comparator model. Otherwise, our model has a greater root mean square error.

We start by comparing with the classical industry benchmark, the Black-Scholes model \citep{black1973pricing}. The parameters of the Black-Scholes model are estimated using maximum likelihood. The results are illustrated in Figure \ref{fig:Black-Scholes}. For each date, the uncalibrated version of the model model gives better results than the Black-Scholes model for at least one proxy. However, there is no proxy that consistently gives better results than the Black-Scholes model. By contrast, the calibrated version of the model consistently gives better results than the Black-Scholes model for all proxies. The performance of all proxies follows a broadly similar pattern across the time period, with the exception of \WVlower, which often performs best before calibration, only for other proxies to catch up during the calibration process. One possible explanation for the poor performance of \WVlower\ on 1 March 2022 is that the invasion of Ukraine on 24 February caused a market shock and, in particular, that it had an immediate effect on option prices that could not be explained by the delayed pre-invasion values of \WVlower; the delay parameter for this model was 13 days (see Figure~\ref{fig:InferenceResults}).

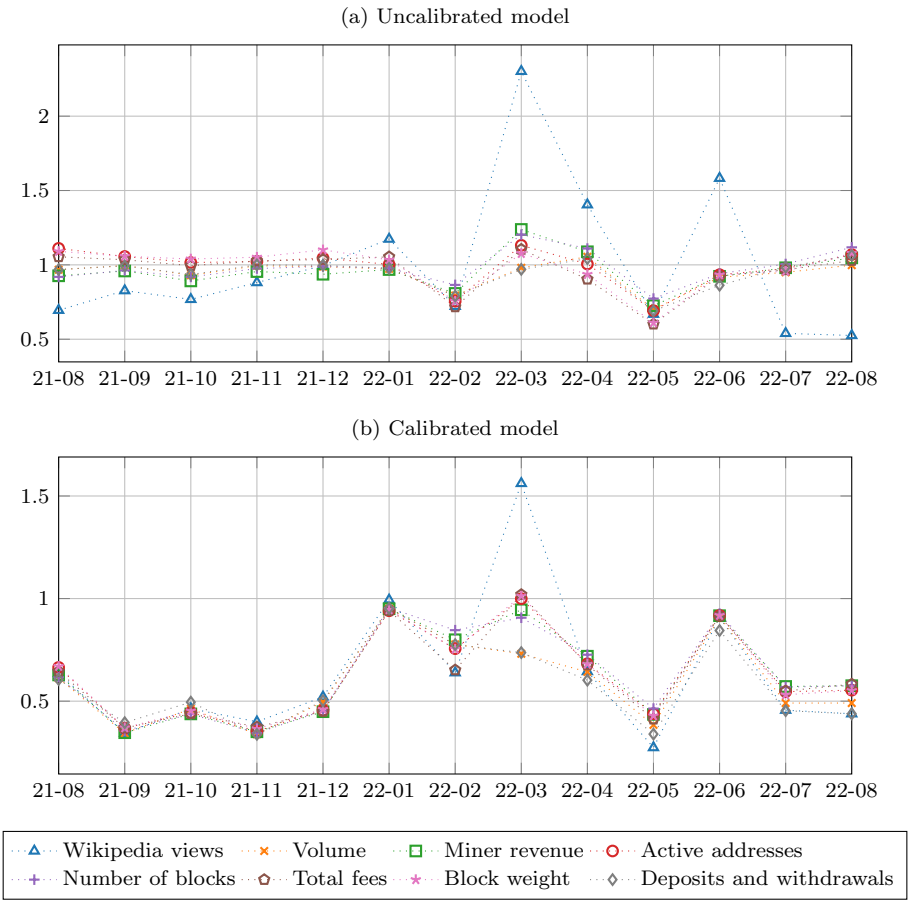
\begin{figure}
\centering
\begin{tikzpicture}
   \pgfplotstableread[col sep=comma,]{Black-Scholes-uncalibrated.csv} {\data}
	\begin{axis} [RMSEstyle,name=uncalibrated, legend style={legend columns=4, legend cell align=left}, legend to name=leg:8-proxies]
	    \addplot+ table [x expr={\coordindex}, y={Wikipedia Views}] {\data}; \addlegendentry{\WV}
        \addplot+ table [x expr={\coordindex}, y={Log Volume}] {\data}; \addlegendentry{\LV}
        \addplot+ table [x expr={\coordindex}, y={Log Miner Revenue (USD)}] {\data}; \addlegendentry{\MR}
        \addplot+ table [x expr={\coordindex}, y={Unique Adresses}] {\data}; \addlegendentry{\AAdd}
	    \addplot+ table [x expr={\coordindex}, y={Block Count}] {\data}; \addlegendentry{\NB}
        \addplot+ table [x expr={\coordindex}, y={Log Fees}] {\data}; \addlegendentry{\TF}
        \addplot+ table [x expr={\coordindex}, y={Block Sum Weigh}] {\data}; \addlegendentry{\BW}
        \addplot+ table [x expr={\coordindex}, y={Log 0.5*(Deposit + Withdrawls)}] {\data}; \addlegendentry{\DW}
     \end{axis}
	\node[align=center,anchor=north] at ($(uncalibrated.north)+(0,0.5\vSpace)$) {(a) Uncalibrated model};
	
   \pgfplotstableread[col sep=comma,]{Black-Scholes-calibrated.csv} {\data}
	\begin{axis} [RMSEstyle , name=calibrated, anchor=north, at = {($(uncalibrated.south)+(0,-\vSpace)$)}]
	    \addplot+ table [x expr={\coordindex}, y={Wikipedia Views}] {\data};
        \addplot+ table [x expr={\coordindex}, y={Log Volume}] {\data};
        \addplot+ table [x expr={\coordindex}, y={Log Miner Revenue (USD)}] {\data};
        \addplot+ table [x expr={\coordindex}, y={Unique Adresses}] {\data};
	    \addplot+ table [x expr={\coordindex}, y={Block Count}] {\data};
        \addplot+ table [x expr={\coordindex}, y={Log Fees}] {\data};
        \addplot+ table [x expr={\coordindex}, y={Block Sum Weigh}] {\data};
        \addplot+ table [x expr={\coordindex}, y={Log 0.5*(Deposit + Withdrawls)}] {\data};
     \end{axis}
	\node[align=center,anchor=north] at ($(calibrated.north)+(0,0.5\vSpace)$) {(b) Calibrated model};
	
	\node[align=center,anchor=north] at ($(calibrated.south)+(0,-0.5\vSpace)$) {\pgfplotslegendfromname{leg:8-proxies}};
\end{tikzpicture}

\caption{Root mean square error of our model (calibrated and uncalibrated versions) divided by root mean square error of Black-Scholes model}
\label{fig:Black-Scholes}
\end{figure}

A second natural comparator can be found in the Heston model \citep{heston1993closed}, which has very similar dynamics, the only differences being that the Heston model has possible correlation, no delay, and a volatility process that is unobserved. It is possible to calibrate the Heston model in its entirety under the pricing measure, in which case it substantially outperforms our model (among many others) \citep[cf.][]{madan2019advanced}. In order to get a more like-for-like comparison, we estimate the parameters of the Heston model from historical Bitcoin prices as well as the historical 30 days variance of log-returns, using the R package Yuima \citep{YUIMA}. Figure \ref{fig:Heston}(a) shows that the uncalibrated version of our model generally holds its own against the Heston model estimated in this way; there is no clear distinction between the two models. The calibrated version performs much better against the Heston model (except for April and May 2022; there is no immediate explanation why the Heston model should do so much better in this period).

\begin{figure}
\centering
\begin{tikzpicture}
   \pgfplotstableread[col sep=comma,]{Heston-uncalibrated.csv} {\data}
	\begin{axis} [RMSEstyle,name=uncalibrated]
	    \addplot+ table [x expr={\coordindex}, y={Wikipedia Views}] {\data};
        \addplot+ table [x expr={\coordindex}, y={Log Volume}] {\data};
        \addplot+ table [x expr={\coordindex}, y={Log Miner Revenue (USD)}] {\data};
        \addplot+ table [x expr={\coordindex}, y={Unique Adresses}] {\data};
	    \addplot+ table [x expr={\coordindex}, y={Block Count}] {\data};
        \addplot+ table [x expr={\coordindex}, y={Log Fees}] {\data};
        \addplot+ table [x expr={\coordindex}, y={Block Sum Weigh}] {\data};
        \addplot+ table [x expr={\coordindex}, y={Log 0.5*(Deposit + Withdrawls)}] {\data};
     \end{axis}
	\node[align=center,anchor=north] at ($(uncalibrated.north)+(0,0.5\vSpace)$) {(a) Uncalibrated model};
	
   \pgfplotstableread[col sep=comma,]{Heston-calibrated.csv} {\data}
	\begin{axis} [RMSEstyle , name=calibrated, anchor=north, at = {($(uncalibrated.south)+(0,-\vSpace)$)}]
	    \addplot+ table [x expr={\coordindex}, y={Wikipedia Views}] {\data};
        \addplot+ table [x expr={\coordindex}, y={Log Volume}] {\data};
        \addplot+ table [x expr={\coordindex}, y={Log Miner Revenue (USD)}] {\data};
        \addplot+ table [x expr={\coordindex}, y={Unique Adresses}] {\data};
	    \addplot+ table [x expr={\coordindex}, y={Block Count}] {\data};
        \addplot+ table [x expr={\coordindex}, y={Log Fees}] {\data};
        \addplot+ table [x expr={\coordindex}, y={Block Sum Weigh}] {\data};
        \addplot+ table [x expr={\coordindex}, y={Log 0.5*(Deposit + Withdrawls)}] {\data};
     \end{axis}
	\node[align=center,anchor=north] at ($(calibrated.north)+(0,0.5\vSpace)$) {(b) Calibrated model};
	
	\node[align=center,anchor=north] at ($(calibrated.south)+(0,-0.5\vSpace)$) {\pgfplotslegendfromname{leg:8-proxies}};
\end{tikzpicture}

\caption{Root mean square error of our model divided by root mean square error of Heston model calibrated to historical data}
\label{fig:Heston}
\end{figure}
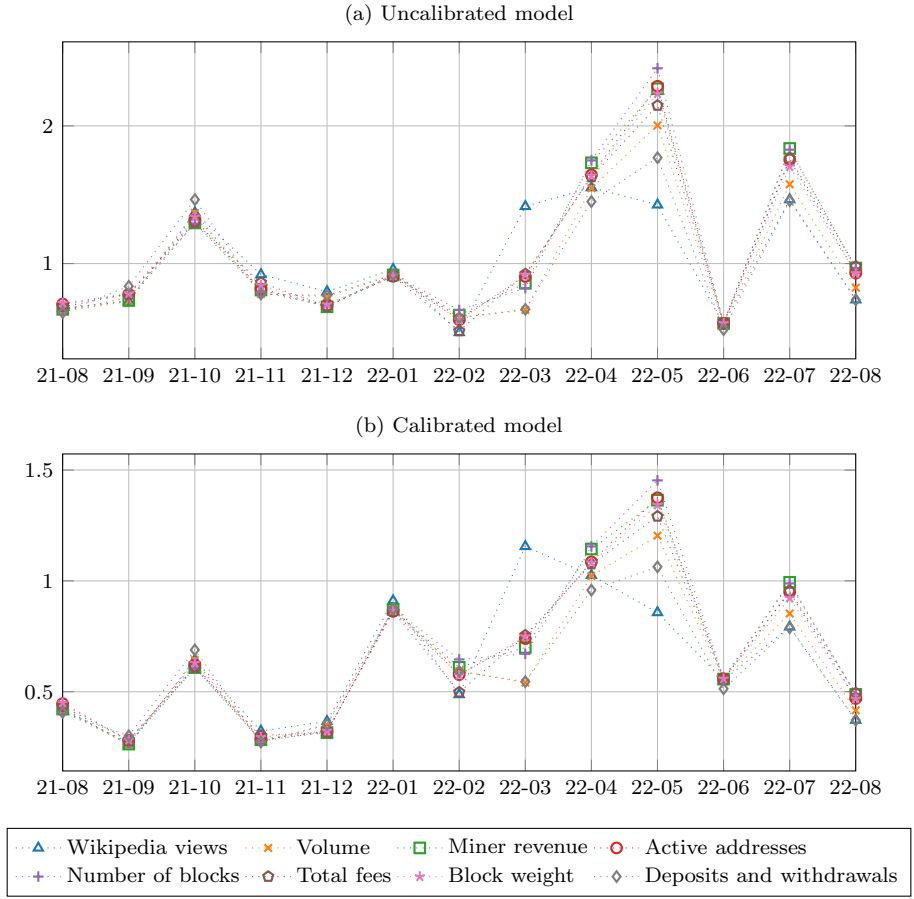

\cite{Cretarola_FigaTalamanca_Patacca2020} recently proposed an attention-based model (accompanied by a specialised estimation procedure) that is similar to ours, the main differences being that it uses geometric Brownian motion to model market attention, and that market attention is assumed to also affect the drift of Bitcoin prices under the real-world probability. Similar to ours, this model has an uncalibrated and a calibrated version in that there is one parameter that can be used to select the pricing measure; we do this by minimising the root mean square error of option prices. The results are presented in Figure \ref{fig:Cretarola}, which compares the root mean square error of the two models in both their non-calibrated and calibrated versions. The \cite{Cretarola_FigaTalamanca_Patacca2020} model was applied to all proxies, including those where the Kolmogorov-Smirnov test suggests that the geometric Brownian motion is not appropriate (see Table \ref{table:proxies-info} and discussion).  As can be seen in Figure~\ref{fig:Cretarola}(a), when comparing uncalibrated versions of the models, the Cretarola-Fig\`a-Talamanca-Patacca model outperforms our model in one-year periods up to January 2022, but our model performs better (often significantly better for some proxies) in later periods. One explanation for this might be that during the early part of the time window, many proxies for market attention were rapidly increasing (see Figure \ref{fig:proxies}), and hence more closely resembled geometric Brownian motion than a mean-reverting process, whereas stability in the later part meant that a mean-reverting process fits better. Comparing the calibrated versions of both models (see Figure \ref{fig:Cretarola}(b)), our model generally performs better for the proxies studied in this paper, often significantly better in those cases where it did well before calibration. This may be due to the fact that our model is calibrated using two parameters rather than the single parameter of the Cretarola-Fig\`a-Talamanca-Patacca model.

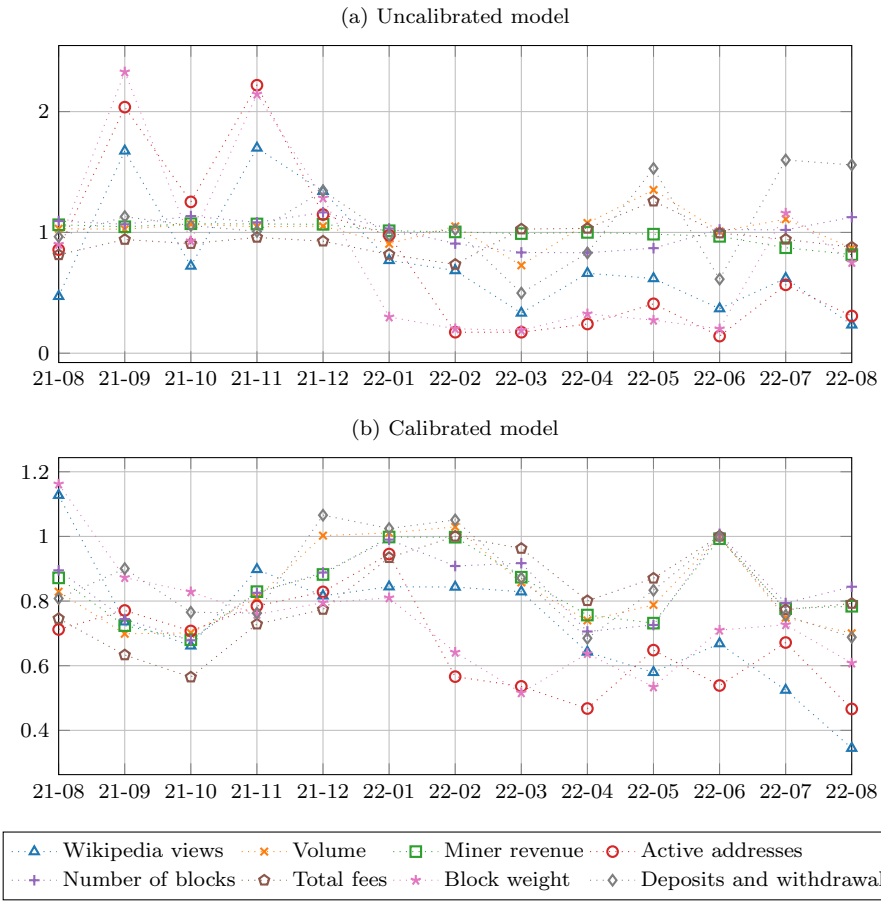
\begin{figure}
\centering
\begin{tikzpicture}
   \pgfplotstableread[col sep=comma,]{Cretarola-uncalibrated.csv} {\data}
	\begin{axis} [RMSEstyle,name=uncalibrated]
	    \addplot+ table [x expr={\coordindex}, y={Wikipedia Views}] {\data};
        \addplot+ table [x expr={\coordindex}, y={Log Volume}] {\data};
        \addplot+ table [x expr={\coordindex}, y={Log Miner Revenue (USD)}] {\data};
        \addplot+ table [x expr={\coordindex}, y={Unique Adresses}] {\data};
	    \addplot+ table [x expr={\coordindex}, y={Block Count}] {\data};
        \addplot+ table [x expr={\coordindex}, y={Log Fees}] {\data};
        \addplot+ table [x expr={\coordindex}, y={Block Sum Weigh}] {\data};
        \addplot+ table [x expr={\coordindex}, y={Log 0.5*(Deposit + Withdrawls)}] {\data};
     \end{axis}
	\node[align=center,anchor=north] at ($(uncalibrated.north)+(0,0.5\vSpace)$) {(a) Uncalibrated model};
	
   \pgfplotstableread[col sep=comma,]{Cretarola-calibrated.csv} {\data}
	\begin{axis} [RMSEstyle , name=calibrated, anchor=north, at = {($(uncalibrated.south)+(0,-\vSpace)$)}]
	    \addplot+ table [x expr={\coordindex}, y={Wikipedia Views}] {\data};
        \addplot+ table [x expr={\coordindex}, y={Log Volume}] {\data};
        \addplot+ table [x expr={\coordindex}, y={Log Miner Revenue (USD)}] {\data};
        \addplot+ table [x expr={\coordindex}, y={Unique Adresses}] {\data};
	    \addplot+ table [x expr={\coordindex}, y={Block Count}] {\data};
        \addplot+ table [x expr={\coordindex}, y={Log Fees}] {\data};
        \addplot+ table [x expr={\coordindex}, y={Block Sum Weigh}] {\data};
        \addplot+ table [x expr={\coordindex}, y={Log 0.5*(Deposit + Withdrawls)}] {\data};
     \end{axis}
	\node[align=center,anchor=north] at ($(calibrated.north)+(0,0.5\vSpace)$) {(b) Calibrated model};
	
	\node[align=center,anchor=north] at ($(calibrated.south)+(0,-0.5\vSpace)$) {\pgfplotslegendfromname{leg:8-proxies}};
\end{tikzpicture}

\caption{Root mean square error of our model divided by root mean square error of Cretarola-Fig\`a-Talamanca-Patacca model}
\label{fig:Cretarola}
\end{figure}

\section{Conclusions and outlook} \label{sec:conclusion}

In this paper, we proposed a model for pricing Bitcoin options in which the volatility is proportional to the market attention, subject to a small delay. We showed how the parameters of this model can be estimated and furthermore calibrated to option pricing data in order to obtain a risk-neutral measure. The conditional characteristic function of the log returns in this model admit a formula in closed form, which makes it possible to price options using well-established Fourier transform methods. The model can be used with a range of proxies for the market attention, and it compares well to (and generally outperforms) the Black-Scholes model, the Heston model when calibrated to historic volatility, and the attention-based Cretarola-Fig\`a-Talamanca-Patacca model.

There are a number of directions in which the model can be generalised, all of which are part of our ongoing research plans. For example, the Brownian motions driving the interest and price processes are assumed to be independent in this paper. This can be ameliorated by either introducing some direct feedback between the interest and price processes, or by allowing the Brownian motions to be correlated. Such extensions can, however, endanger the Markov property, and hence introduce new technical challenges, and it may well prove impossible to derive a closed-form conditional characteristic function, which is a strength of the current model. A further modification could be made by allowing the interest process to affect the drift of the price process---this would potentially alter the technical details of the change-of-measure.

The data (see Figure \ref{fig:InferenceResults}) suggests that different proxies for the market attention affect Bitcoin prices with different delays. This suggests that it may be worth constructing a multi-factor model, perhaps similar to the double Heston model \citep{christoffersen2009shape}, in order to accommodate multiple interest processes.

Finally, more advanced methodologies, such as jumps, time changed processes, etc, may prove fruitful in order to better capture the distribution of Bitcoin returns. A number of such possibilities have already been studied by \cite{GuineaJulia2022}, who showed that it is possible to obtain at least semi-closed formulae that can be used to price vanilla European options.

\bibliographystyle{tfcad}
\bibliography{interactcadsample}

\end{document}